\documentclass[11pt,reqno]{amsart}
\usepackage{indentfirst}
\usepackage{eurosym}
\usepackage{amsmath,amssymb,epsfig}
\usepackage{amsfonts, lscape, harvard}
\usepackage{longtable,colortbl,array}
\usepackage{pdflscape}
\usepackage{subfig}
\usepackage{caption}
\usepackage[
  bookmarks=true,
  bookmarksopen=true,
  breaklinks=true,
  colorlinks=true,
  linkcolor=magenta(dye),
  citecolor=magenta(dye),
]{hyperref}
\usepackage[parfill]{parskip}
\usepackage{floatrow}
\usepackage{booktabs,rotating}
\usepackage{graphicx}
\usepackage{enumitem}

\definecolor{lust}{rgb}{0.9, 0.13, 0.13}
\definecolor{magenta(dye)}{rgb}{0.79, 0.08, 0.48}
\theoremstyle{plain}
\newtheorem{lemma}{Lemma}

\newtheorem{algorithm}{Algorithm}

\newtheorem{assumption}{Assumption}

\newtheorem{claim}{Claim}

\newtheorem{corollary}{Corollary}

\newtheorem{definition}{Definition}
\newtheorem{example}{Example}

\newtheorem{theorem}{Theorem}

\numberwithin{equation}{section}
\theoremstyle{definition}

\DeclareMathOperator*{\argmax}{arg\,max}

\let\P\relax
\newcommand{\P}{\mathbf{P}}
\newcommand{\Q}{\mathbf{Q}}
\newcommand{\R}{\mathbb{R}}

\newcommand{\X}{\mathcal{X}}
\newcommand{\Y}{\mathcal{Y}}

\newcommand{\I}{\mathcal{I}}
\newcommand{\J}{\mathcal{J}}

\input{tcilatex}

\begin{document}
\title[Stable Matching with Money Burning]{Aggregate Stable Matching with Money Burning}
\author[Alfred Galichon]{Alfred Galichon${}^\S{}^\dagger$}
\author[Yu-Wei Hsieh]{Yu-Wei Hsieh$^\clubsuit$}
\author[Antoine Jacquet]{Antoine Jacquet$^\dagger$}
\date{\today.\\
\indent \quad $^\S$ New York University 
and Sciences Po.
Email: ag133@nyu.edu\\
\indent \quad $^\clubsuit$ Amazon.
Email: yuweihsieh01@gmail.com\\
\indent \quad $^\dagger$ Sciences Po, Department of Economics.
Email: antoine.jacquet@sciencespo.fr \\
[5pt] \textbf{Acknowledgement:}
Galichon acknowledges support from the NSF grant
DMS-1716489, and he and Jacquet acknowledge support from the European Research Council, Grant CoG-866274. This paper has benefited from conversations with Federico Echenique,
Christopher Flinn, Jeremy Fox, Douglas Gale, Bryan Graham, Guy
Laroque, Thierry Magnac, Charles Manski, Konrad Menzel, Larry Samuelson,
Simon Weber, Glen Weyl, and comments from seminar participants at the
Toulouse School of Economics, the Fields Institute for Mathematical
Sciences, Carnegie Mellon University, Tepper School of Business, CalTech,
UNC, Columbia University, Yale University, Celebrating Chris Flinn's 65th
Birthday Conference, and the 2015 California Econometrics Conference.
Hsieh's contribution to the paper reflects the work completed prior to his joining Amazon.
This paper is dedicated to the memory of YingHua He, a dear friend and colleague whose scholarship has left a lasting mark on the field.
We are honored to contribute to this special issue in his memory.}

\begin{abstract}
We propose an aggregate notion of non-transferable utility (NTU) stability for decentralized matching markets with fixed prices, where market clearing is achieved through one-sided money burning, which can be interpreted as waiting.
Agents are grouped into observable types and are indifferent among individuals within type; equilibrium is defined at the type level and delivers equal indirect utility within each type.
We introduce money burning into two types of NTU models:
In a deterministic model, we relate our notion to classical Gale--Shapley stability and show how money burning decentralizes stable outcomes under aggregation.
We then introduce separable random utility, obtaining an NTU counterpart to Choo and Siow (2006).
We prove the existence and uniqueness of equilibrium and provide a stationary queueing interpretation.
Finally, we develop a generalized deferred acceptance algorithm based on alternating constrained discrete-choice problems and prove its convergence to the unique equilibrium.
\end{abstract}

\vspace*{-1.5cm}
\maketitle


{\footnotesize \textbf{JEL Classification}: C78, D58\vskip50pt }


\section{Introduction}
\setcounter{page}{1}
\setcounter{equation}{0}

The literature on matching markets typically distinguishes between models with transferable utility (TU), in which a numéraire good---often money---clears the market,
and models with non-transferable utility (NTU), in which no numéraire exchange is possible.
Traditionally, models with transfers have been applied to \emph{decentralized} markets such as labor, housing, or marriage markets,
whereas models without transfers have typically been used to study \emph{centralized} markets such as school assignments, organ transplants, or medical residents, where a market designer clears the market using an algorithm without prices.
In the real world, however, decentralized matching markets without transfers are common due to sticky prices, capacity constraints, or regulatory requirements. In healthcare, for instance, excess demand is not cleared through transferable surge pricing: bed capacity is fixed in the short term, and copayments are predetermined by the insurance policy. Similarly, in traditional taxi markets, unit fares are set by regulators, and no centralized assignment algorithm exists.

In a decentralized market, an equilibrium typically requires a market-clearing device: a variable that adjusts the utilities attached to each option until every agent gets their preferred option given these adjustments.
In matching markets with TU, this role may be played by prices or intra-match bargaining.
In traditional NTU settings, however, partners' utilities cannot be adjusted through transfers, so the market cannot be cleared.
How then are stable outcomes selected and enforced when multiple agents want the same partner and ties cannot be cleanly broken?
In practice, such conflicts can be resolved through negative externalities tied to overdemanded options---that is, wasteful competition in the form of waiting lines, costly applications such as college admission essays, overinvestment in quality signals, or other forms of congestion, which we refer to under the umbrella term \emph{money burning}.
Waiting lines, for example, are then an institutionalized tie-breaking device: instead of fighting or engaging in cutthroat competition, agents compete by burning time, which lowers the effective payoff from popular matches until demand is consistent with capacity.

We thus propose a decentralized NTU matching model in which a money-burning mechanism---interpreted as time---replaces the price as the bidding device.
Our approach is \emph{competitive} in the sense that waiting lines form in front of over-demanded agents and reduce the payoff from those matches in proportion to the time waited,
yet it remains \emph{non-transferable utility}: time cannot be transferred across sides and is pure dissipation.
Our equilibrium notion is \emph{aggregate}:
Agents are partitioned into observable types and are indifferent between individuals of the same type; the equilibrium specifies matching patterns and waiting times (equivalently, payoffs) at the type level.
This delivers an equal-treatment property within types: two identical individuals obtain the same equilibrium payoff, with within-type differences absorbed by money burning.

Consider the following motivating example comparing our equilibrium notion with the classical TU and NTU stable matching.
Suppose that there are two passengers and one taxi. Passenger 1 and passenger 2 value the ride service by 2 and 1, respectively.
The taxi is indifferent between the two passengers, and we normalize its utility to zero.
The \emph{aggregate} outcome is that one passenger gets the ride since there is only one taxi available.
However, different matching models have different implications for welfare and \emph{individual} assignments.
Under the classical NTU model, assigning either passenger 1 or 2 to the taxi yields a stable matching, and NTU stability provides no further guidance.
Under our framework, passenger 1 waits until passenger 2 drops out.
The allocatively efficient assignment is decentralized by wasteful competition on the demand side: the ride goes to passenger 1, who values it most, but one unit of utility is dissipated through waiting, so the social surplus is $1$.
By contrast, under a free-market TU model, the efficient allocation is decentralized by a price of $1$ transferred to the taxi, and the social surplus is $2$.%
\footnote{Surge pricing in ride-hailing can be viewed as a move from NTU toward a TU-like mechanism that reduces waiting.}

We develop this notion in three steps.
First, we introduce \emph{aggregate stable matching with money burning} under deterministic utilities.
We then relate it to the classical individual-level definition of NTU stability \citeasnoun{GaleShapley1962}.
When there is one individual per type, any classical stable matching is also an aggregate stable matching supported by no money burning.
With multiple indistinguishable individuals per type, however, classical stability alone does not specify how overdemanded matches are decentralized while preserving equal treatment within type. Money burning provides exactly this missing equilibrium device.
Our contribution is therefore not to replace classical stability, but to provide its decentralized aggregate implementation under equal treatment within type.

Second, we extend the model to separable random utility à la \citeasnoun{ChooSiow2006}.
Waiting times enter discrete-choice demand exactly like prices, and equilibrium requires market clearing together with one-sided money burning in every segment.
Under a mild continuity assumption on taste shocks, we prove that there exists a \emph{unique} aggregate stable matching with money burning.
We also provide a stationary dynamic interpretation in which the static equilibrium coincides with the steady state of a market-clearing system with queues.

Third, we propose a generalized deferred acceptance algorithm for the random-utility model.
The key step is to view deferred acceptance as alternating between two constrained discrete-choice problems, with waiting times acting as shadow prices that enforce capacity constraints.
We prove convergence of the procedure and show that its limit coincides with the unique aggregate stable matching with money burning.

\subsection*{Related literature.}

This paper is related to four streams of the economic literature:
(i) non-price rationing,
(ii) decentralized matching without transfers,
(iii) deferred acceptance algorithms, and
(iv) matching with unobservable heterogeneity.
First, \emph{non-price rationing} arises in many diverse situations such as sticky prices in the macroeconomic theory of disequilibrium (e.g., \citeasnoun{Benassy1976}, \citeasnoun{GourierouxLaroque1985}, \citeasnoun{Dreze1987});
in credit rationing (e.g., \citeasnoun{Sealy1979});
in housing market with rent control (e.g., \citeasnoun{GlaeserLuttmer2003});
in mechanism design with money burning (e.g., \citeasnoun{HartlineRoughgarden2008}, \citeasnoun{Braverman_etal_2016});
and in health economics (e.g., \citeasnoun{LindsayFeigenbaum1984}, \citeasnoun{Iversen1993}, \citeasnoun{MartinSmith1999}, \citeasnoun{IversenSiciliani2011}).
The mathematical theory of queuing is surveyed in \citeasnoun{HassinHaviv2003}.
\citeasnoun{Condorelli2012} studies other forms of non-price rationing, such as priority lists and lotteries,
from the perspective of mechanism design.
In econometrics, simultaneous demand--supply systems subject to the quantity rationing constraints have been studied for example by \citeasnoun{FairJaffee1972},
\citeasnoun{GourierouxLaffontMonfort1980}, and \citeasnoun{Maddala1986}.
Beyond economics, there is a controversy about the social desirability of waiting lines as a rationing mechanism;
a vocal advocate in favor of them is Michael Sandel \citeyear{Sandel2013}.

Second, there is a large literature on ``market design problems'' focused on centralized matching models without transfers, which we will not review here; we shall focus instead on the narrower literature on \emph{decentralized matching without transfers}.
Our basic observation is that it is extremely difficult to define the aggregate stable matching when agents are clustered into types of indistinguishable individuals.
Indeed, in the absence of transfers, it can be challenging to break ties between identical individuals,%
\footnote{A literature on fractional stable matchings was initiated with the interesting paper of \citeasnoun{RothRothblumVate1993}; however, this model was not designed to handle aggregation problems.}
and it may therefore be difficult to enforce the desirable requirement that two agents with similar
characteristics will obtain the same payoff at equilibrium.
Models in the literature have resolved this difficulty mostly by pursuing two approaches.
The first approach involves stochastic rationing (see \citeasnoun{Gale1996} and
references therein) or the introduction of search frictions (e.g., \citeasnoun{BurdettColes1997}, \citeasnoun{Smith2006} and the references therein).
Search frictions provide a way to stochastically ration demand and supply and a rationale to explain variations in the equilibrium payoffs of similar individuals.
The second approach involves the introduction of heterogeneity, which can either be observed \cite{AzevedoLeshno2016} or unobserved, and can be captured in a random utility model (see \citeasnoun{Dagsvik2000} and \citeasnoun{Menzel2015}, who use logit heterogeneities).
\citeasnoun{CheKoh2016} have investigated the case of decentralized college admission with uncertain student preferences.
In particular, writing college-specific essays can be viewed as a money burning mechanism.
\citeasnoun{EcheniqueYariv2013} and \citeasnoun{NiederleYariv2009} provide other approaches to study decentralized matching markets.
\citeasnoun{EcheniqueLeeShumYenmex} offer a characterization of rationalizability of matchings without transfers in the spirit of revealed preference.

Third, the algorithm studied in this paper belongs to the broad class of deferred acceptance procedures initiated by \citeasnoun{GaleShapley1962}, extended to markets with flexible transfers by \citeasnoun{KelsoCrawford1982}, and generalized to matching with contracts by \citeasnoun{HatfieldMilgrom2005}.
Those papers study finite-agent markets, typically many-to-one and with salaries or contracts, whereas we study a one-to-one, fixed-price NTU environment with a continuum of agents aggregated into observable types.
In our setting, overdemand is resolved through endogenous money burning rather than through salary or contract terms. \citeasnoun{Galichon2026}, which also reviews an earlier version of this paper, reinterprets \citeasnoun{HatfieldMilgrom2005} through the lens of the symmetric constrained-choice framework developed here. In addition, \citeasnoun{GalichonHsiehSylvestre2024} extends our deferred acceptance construction to the demand-correspondence case in order to handle ties, introducing for this purpose a novel theory of monotone comparative statics.

Fourth, we consider models with stochastic utility components.
Therefore, our paper can be seen as the separable NTU counterpart of the separable TU model
with random utility proposed by \citeasnoun{GalichonSalanie2022}, who extend
the approach of \citeasnoun{ChooSiow2006} beyond the logit case.
\citeasnoun{GalichonKominersWeber2019} show that by choosing a suitable specification, our model arises
naturally as the limiting case of imperfectly transferable utility models with random utility.%
\footnote{However, \citeasnoun{GalichonKominersWeber2019} do not study how the NTU matching can be decentralized and the micro theory of money burning.}
As described by \citeasnoun{AzevedoLeshno2016}, our notion of equilibrium can be interpreted as
the solution of a tâtonnement process in a demand and supply framework;
however, in contrast to their framework, ours accommodates a finite number of agents and does not require consideration of a continuous limit.

\medskip

The rest of the paper is organized as follows.
Section \ref{sec:matching_deterministic_utility} introduces aggregate stability with money burning under deterministic utilities and relates it to classical NTU stability.
Section \ref{sec:matching_random_utility} extends the analysis to random utility, establishes existence and uniqueness, and provides a stationary queuing interpretation in a dynamic framework.
Section \ref{sec:deferred_acceptance} presents a generalized deferred acceptance algorithm based on constrained discrete choice and proves its convergence.
Section \ref{sec:conclusion} concludes with a brief discussion of the econometric implications of our equilibrium notion.
Appendices collect technical results and proofs.

\section{Aggregate Stable Matching: the Case of Deterministic Utility}
\label{sec:matching_deterministic_utility}

\subsection{Motivation and Definition}

We consider the problem of matching different types of taxis with different types of passengers.
There are $n_x$ passengers of type $x \in \X$, where $x$ includes the pick-up location, the size of the party, the type of vehicle requested, etc.
There are $m_y$ taxis of type $y \in \Y$ available, where $y$ includes the service offered (e.g., Pool, SUV, or Limo), amenities (e.g., video screen or snack box), driver rating, etc.
For each type, the number $n_x$ or $m_y$ is a strictly positive integer.
Agents are assumed to have preferences over types;
they are indifferent between two agents of the same
type.
A type-$x$ passenger enjoys utility $\alpha_{xy}$ from traveling in a taxi of type $y$, and a type-$y$ taxi enjoys $\gamma_{xy}$ from serving a passenger of type $x$.
The outside option is labeled by $0$, and the corresponding reservation utility of both passengers and taxis is normalized to zero without loss of generality.

We consider a non-transferable utility setup, namely, the price is predetermined---a common practice in the taxi industry.
When the market-clearing price is absent, the demand may not be equal to the supply, leading to quantity rationing or waiting lines that serve as the market clearing device.
We assume that utility is \emph{quasi-linear} in the time waited:
The utility obtained by a passenger $x$ riding in a taxi $y$ after waiting an amount of time $\tau_{xy}^\alpha \geq 0$ is $\alpha_{xy} - \tau_{xy}^\alpha$,
while the utility for a taxi $y$ transporting a passenger $x$ after waiting $\tau_{xy}^\gamma \geq 0$ is $\gamma_{xy} - \tau_{xy}^\gamma$.%
\footnote{Normalizing the marginal disutility of waiting time to one is without loss of generality, since utilities can be expressed in equivalent time units.
Instead, one could also model time as a discount factor that decreases utility.
However, this paradigm will lead to a non-quasilinear model.
It is possible to extend our analysis to such a case by utilizing abstract convex analysis as in \citeasnoun{Bonnet2020}.
On the other hand, time waited is type-specific in our model, which is akin to the type-specific price in \citeasnoun{ChooSiow2006}.
In the case of taxis, it may take longer for passengers in certain locations to match with taxis.
\citeasnoun{GalichonHsieh_hedonic} study a hedonic model with a common waiting time.}
Either passengers or taxis have to wait, depending on which side of the market is in shortage.
In a frictionless market,%
\footnote{In the real world, both taxis and passengers may incur non-zero waiting time before the driver reach the pick-up location.
Our model abstracts from this friction and only considers the ``net'' waiting time.}
there cannot exist simultaneously a nonempty waiting line of both passengers and taxis in the market segment $xy$:
\begin{equation*}
\min \big\{ \tau_{xy}^\alpha, \tau_{xy}^\gamma \big\} = 0.
\end{equation*}

An aggregate matching is a matrix $\mu = (\mu_{xy})$ where $\mu_{xy}$ denotes the number of passengers of type $x$ riding in taxis of type $y$.
We consider the competitive equilibrium in which passengers choose the type of taxis that maximize their surplus, and taxis choose the type of passengers that maximize their surplus.%
\footnote{%
In our model of a decentralized NTU matching market, taxis also play an active role in selecting passengers.
Traditionally, when passengers book a ride, the dispatch center broadcasts through its network to reach nearby drivers, and the one who response first wins the ride.
Today, this practice is largely conducted through mobile apps.
For example, Uber drivers can express their preferences over destinations for two trips everyday.
The assignment algorithm will attempt to match them first with passengers who request similar destinations.
Our behavioral assumption attempts to capture the fact that the drivers still have certain freedom to select the type of passengers \emph{before} boarding, and it does not contradict the common practice that the taxis cannot refuse to serve according to the destination \emph{after} the consumers are on board.
}
%
Let $u_x$ and $v_y$ be the indirect utilities of type-$x$ passengers and type-$y$ taxis, respectively.
We have
\begin{equation*}
u_x = \max_{y \in \Y}
\big\{ \alpha_{xy} - \tau_{xy}^\alpha, 0 \big\}
\quad\text{and}\quad
v_y = \max_{x \in \X}
\big\{ \gamma_{xy} - \tau_{xy}^\gamma, 0 \big\}.
\end{equation*}
Therefore, $u_x \geq \alpha_{xy} - \tau_{xy}^\alpha$ with equality if $x$ chooses $y$, i.e., if $\mu_{xy} > 0$.
Similarly, $v_y \geq \gamma_{xy} - \tau_{xy}^\gamma$ with equality if $\mu_{xy} > 0$.
As a result,
\begin{equation*}
\max \big\{ u_x - \alpha_{xy}, v_y - \gamma_{xy} \big\}
\geq \max \big\{ -\tau_{xy}^\alpha, -\tau_{xy}^\gamma \big\}
= -\min \big\{ \tau_{xy}^\alpha, \tau_{xy}^\gamma \big\} = 0,
\end{equation*}
and if $\mu_{xy} > 0$ this also holds with equality.
This brings us to the following definition of an aggregate stable matching:

\begin{definition}
\label{def:stable_matching_deterministic}
In the deterministic utility case, $(\mu,u,v)$ is an \emph{aggregate stable matching with money burning} if it meets the following six conditions:
\begin{enumerate}[label=(\roman*), topsep= 0pt, itemsep= 0pt]

\item $\mu_{xy}$ is an integer for all $xy \in \X \times \Y$,

\item $\sum_{y \in \Y} \mu_{xy} \leq n_x$ for all $x \in \X$,

\item $\sum_{x \in \X} \mu_{xy} \leq m_y$ for all $y \in \Y$,

\item for all $x \in \X$ and $y \in \Y$, $\max \big\{ u_x - \alpha_{xy}, v_y - \gamma_{xy} \big\} \geq 0$ with equality if $\mu_{xy}>0$,

\item for all $x \in \X$, $u_x \geq 0$ with equality if
$\mu_{x0} := n_x - \sum_{y \in \Y} \mu_{xy} > 0$,

\item for all $y \in \Y$, $v_y \geq 0$ with equality if
$\mu_{0y} := m_y - \sum_{x \in \X} \mu_{xy} > 0$.
\end{enumerate}

\end{definition}

Conditions (i)--(iii) are standard feasibility constraints on the matching.
Conditions (v)--(vi) enforce equal treatment for types containing unmatched agents.
Finally, condition (iv) is the stability condition which enforces one-sided waiting. 
Given an aggregate stable matching with money burning $(\mu,u,v)$, compatible waiting times can be recovered as $\tau_{xy}^\alpha = \max \{\alpha_{xy} - u_x, 0\}$ and $\tau_{xy}^\gamma = \max\{\gamma_{xy} - v_y, 0\}$.

Our equilibrium notion is distinct from the classical stable matching, e.g., \citeasnoun{GaleShapley1962}, in two fundamental ways.
First, we have introduced waiting lines as a competitive money-burning mechanism to decentralize the stable matching.
In this regard, our approach has a close connection with the transferable utility matching problems studied in
\citeasnoun{Becker1973} and \citeasnoun{ShapleyShubik1972}.
Indeed, by replacing the $\max$ function by the summation function in point (iv) of Definition \ref{def:stable_matching_deterministic}, one obtains the definition of stable matching with transferable utility.
By contrast, the classical matching theory typically relies on a centralized algorithm to achieve a stable matching.
Second, in our setup, passengers only care about the \emph{type} of the service.
Two taxis with distinct license plates but of the same type are perfect substitutes.
As a consequence, our notion of stable matching is an \emph{aggregate} equilibrium.
By contrast, in the classical setup, agents are allowed to express their preference ranking at the individual level.


\subsection{Comparison with Classical Stable Matching}
\label{sec:compNTUStable}

In this section we establish the connection between Definition \ref{def:stable_matching_deterministic} and the classical definition of stable matching with non-transferable utility, which are not equivalent.
Since \citeasnoun{GaleShapley1962} is based on describing the matching problem at the individual level, we first
need to describe individual passengers and taxis.

Let $\I$ denote the set of passengers, $\J$ the set of taxis, $x_i \in \X$ the observable type of passenger $i \in \I$, and $y_j \in \Y$ the observable type of taxi $j \in \J$.
A match between passenger $i$ and taxi $j$ brings utility $\alpha_{ij}$ to the passenger and $\gamma_{ij}$ to the taxi.
The reservation utilities are still normalized to zero, i.e., $\alpha_{i0} = 0$ and $\gamma_{0j} = 0$.
Because agents of the same type share the same preferences and are perfect substitutes for potential partners, they are indistinguishable.
Therefore,
\begin{equation}  \label{eq:aggregated_preferences}
\alpha_{ij} = \alpha_{x_i y_j}
\quad \text{and} \quad
\gamma_{ij} = \gamma_{x_i y_j}.
\end{equation}
A matching at the individual level is a binary matrix $\pi = (\pi_{ij})$ such that $\pi_{ij} = 1$ if $i$ and $j$ are matched, and $\pi_{ij} = 0$ otherwise.
Under the matching $\pi$, passenger $i$ and taxi $j$ respectively enjoy the utilities $u_i^\pi$ and $v_j^\pi$ given by
\begin{equation}
\label{eq:u_i_v_j}
u_i^\pi = \sum_{j \in \J} \pi_{ij} \alpha_{ij}
\quad \text{and} \quad
v_j^\pi = \sum_{i \in \I} \pi_{ij} \gamma_{ij}.
\end{equation}

Below we summarize the classical definition of stable matching.

\begin{definition}
\label{def:classical_stable_matching}
A matching $\pi$ is a \emph{stable matching in the classical sense} if it meets the following six conditions:
\begin{enumerate}[label=(\roman*), topsep= 0pt, itemsep= 0pt]

\item $\pi_{ij} \in \{0,1\}$ for all $ij \in \I \times \J$,

\item $\sum_{j \in \J}\pi_{ij}\leq 1$ for all $i \in \I$,

\item $\sum_{i \in \I}\pi_{ij}\leq 1$ for all $j \in \J$,

\item for all $i \in \I$ and $j \in \J$,
$\max \big\{ u_i^\pi -\alpha_{ij}, v_j^\pi - \gamma_{ij} \big\} \geq 0$,

\item for all $i \in \I$, $u_i^\pi\geq 0$,

\item for all $j \in \J$, $v_j^\pi\geq 0$.

\end{enumerate}

\end{definition}

Our first result establishes the connection between stable matchings in the classical sense and aggregate stable matchings with money burning as introduced in Definition \ref{def:stable_matching_deterministic}.

\begin{theorem}
\label{thm:EquivWithClassicalNotion}
Assume that $\sum_{i \in \I} \mathbf 1(x_i=x) = n_x$ for all $x \in \X$, $\sum_{j \in \J} \mathbf 1(y_j=y) = m_y$ for all $y \in \Y$, and \eqref{eq:aggregated_preferences} holds.
Then:
\begin{enumerate}[label=(\roman*), topsep= 0pt, itemsep= 0pt]

\item If $\pi$ is a stable matching in the classical sense, then letting
\begin{equation}  \label{eq:mu_from_pi}
\mu_{xy} = \sum_{i \in \I, \, j \in \J} \pi_{ij} \mathbf 1(x_i=x) \mathbf 1(y_j=y),
\end{equation}
\begin{equation}  \label{eq:u_v_from_pi}
u_x = \min_{i:x_i=x} u_i^\pi,
\qquad
v_y = \min_{j:y_j=y} v_j^\pi,
\end{equation}
where $u_i^\pi$ and $v_j^\pi$ are defined in \eqref{eq:u_i_v_j},
the outcome $(\mu,u,v)$ is an aggregate stable matching with money burning.

\item Conversely, if $(\mu,u,v)$ is an aggregate
stable matching with money burning, then any matching $\pi$ satisfying conditions (i)--(iii) of Definition \ref{def:classical_stable_matching} and such that \eqref{eq:mu_from_pi} holds is a stable matching in the classical sense.

\end{enumerate}

\end{theorem}

The first part of Theorem \ref{thm:EquivWithClassicalNotion} suggests that one may have to burn an amount of money in order to decentralize a given stable matching in the classical sense.
Suppose that under the stable matching $\pi$, passenger $i$ of type $x$ is matched with taxi $j$ of type $y$.
Then one can interpret $\tau_i^\alpha = u_i^\pi - u_x$ and $\tau_j^\gamma = v_j^\pi - v_y$ as the waiting times in the associated aggregate stable matching $(\mu,u,v)$.
The waiting times are there to ensure that all agents of the same type receive as much utility as the worse-off agent of that type.
The second part of Theorem \ref{thm:EquivWithClassicalNotion} states that any individual-level matching, as long as its aggregate number of matches by type coincides with a given aggregate stable matching with money burning, is also a stable matching in the classical sense.
Lastly, Theorem \ref{thm:EquivWithClassicalNotion} implies the following corollary:


\begin{corollary}  \label{cor:NoTradeEquil}
When there is one individual of each type, any stable matching in the classical sense can be interpreted as an aggregate stable matching supported by no money burning.
\end{corollary}

In the next section we also prove existence of the aggregate stable matching with money burning in the deterministic utility case by taking a vanishing-randomness limit of the random utility model; see Theorem \ref{thm:limit}.



\section{Aggregate Stable Matching: the Case of Random Utility}
\label{sec:matching_random_utility}

While the finite-agent deterministic model is useful for introducing the equilibrium notion and its connection with classical NTU stability, it does not deliver a smooth aggregate demand system.
To study existence and uniqueness, we now introduce idiosyncratic random utility and move to a continuum of agents.

\subsection{Definition.}
We continue to adopt the language of passengers and taxis as in section \ref{sec:matching_deterministic_utility}.
In the same spirit as \citeasnoun{GalichonSalanie2022}, there is now a continuum of agents on the market, with mass $n_x$ of type-$x$ passengers and mass $m_y$ of type-$y$ taxis.
In contrast to the deterministic utility case, passenger $i$ of type $x$ traveling in a taxi $y$ enjoys not only the systematic utility $\alpha_{xy}$, but also an additively separable random utility component $\varepsilon_{iy}$.
The vector $\varepsilon_i = (\varepsilon_{iy})_y$ of these random components follows a distribution $\P_x$ which may depend on $x$.
Similarly, taxi $j$ of type $y$ serving a passenger $x$ enjoys the systematic utility $\gamma_{xy}$, and an additively separable random utility component $\eta_{xj}$.
The vector $\eta_j = (\eta_{xj})_x$ follows a distribution $\Q_y$ which may depend on $y$.
As in the textbook discrete-choice model \cite{McFadden1976}, we assume that each decision maker observes their realization of the random utility component before making the choice.
The economist who studies the resulting demand system, however, only knows the distributions $(\P_x,\Q_y)$.
We make the following assumption on the random utility component:

\begin{assumption}  \label{ass:nonvanishing}
For all $x \in \X$ and $y \in \Y$, $\P_x$ and $\Q_y$ have a nowhere vanishing density.
\end{assumption}

As in section \ref{sec:matching_deterministic_utility}, the agents' systematic utility is also quasi-linear in the amount of time waited.
We denote by $\tau_{xy}^\alpha$ the waiting time for passengers $x$ wishing to match with a taxi $y$, and by $\tau_{xy}^\gamma$ the waiting time for taxis $y$ wishing to match with a passenger $x$.
The demand for taxis $y$ by passengers $x$ as a function of $\tau^\alpha$, denoted $\boldsymbol \mu_{xy}^\alpha (\tau^\alpha)$, is therefore
\begin{equation} \label{eq:demand_passengers}
\boldsymbol \mu_{xy}^\alpha (\tau^\alpha)
= n_x \, \P_x \Big(
y \in \argmax_{y' \in \Y \cup \{0\}} \big\{ \alpha_{xy'} - \tau_{xy'}^\alpha + \varepsilon_{iy'} \big\}
\Big),
\end{equation}
and the demand for passengers $x$ by taxis $y$ as a function of $\tau^\gamma$, denoted $\boldsymbol \mu_{xy}^\gamma (\tau^\gamma)$, is
\begin{equation} \label{eq:demand_taxis}
\boldsymbol \mu_{xy}^\gamma (\tau^\gamma)
= m_y \, \Q_y \Big(
x \in \argmax_{x' \in \X \cup \{0\}} \big\{ \gamma_{x'y} - \tau_{x'y}^\gamma + \eta_{x'j} \big\}
\Big),
\end{equation}
where by convention $\tau_{x0}^\alpha = \tau_{0y}^\gamma = 0$.
These demand functions react to waiting times exactly like standard discrete-choice demands react to prices.
When the waiting time $\tau_{xy}^\alpha$ increases, passengers $x$ find taxis $y$ less attractive, hence their demand $\boldsymbol \mu_{xy}^\alpha (\tau^\alpha)$ for taxis $y$ decreases while their demand $\boldsymbol \mu_{xy'}^\alpha (\tau^\alpha)$ for other taxis $y' \neq y$ or for the outside option $y'=0$ increases.%
\footnote{Under Assumption \ref{ass:nonvanishing}, all these changes are strict.}
Likewise, when $\tau_{xy}^\gamma$ increases, taxis $y$ find passengers $x$ less attractive, so their demand $\boldsymbol \mu_{xy}^\gamma (\tau^\gamma)$ for passengers $x$ decreases while their demand $\boldsymbol \mu_{x'y}^\gamma (\tau^\gamma)$ for other passengers $x' \neq x$ or for the outside option $x'=0$ increases.

Our solution concept is a frictionless, competitive equilibrium analysis à la \citeasnoun{ChooSiow2006}, in which agents choose their most preferred type of match, taking utilities and waiting times as given.
The optimal choices made by all agents collectively determine the equilibrium matching and the level of money-burning.
At equilibrium, (i) demand equals supply, and (ii) there cannot be a pair where both sides of the market burn money, i.e., a passenger of type $x$ waiting for a taxi of type $y$ while a taxi of type $y$ is simultaneously waiting for a passenger of type $x$.
Formally, we define the aggregate stable matching with money burning as follows:

\begin{definition}  \label{def:stable_matching_random}
In the case of random utility, $(\mu, \tau^\alpha, \tau^\gamma)$ is an aggregate stable matching with money burning if it  verifies simultaneously:
\begin{enumerate}[label=(\roman*), topsep= 0pt, itemsep= 0pt]

\item \textbf{Market Clearing:}
The number of matches of type $xy$ equals both the demand for taxis $y$ by passengers $x$ under $\tau^\alpha$, and the demand for passengers $x$ by taxis $y$ under $\tau^\gamma$; namely,
\begin{equation}  \label{eq:market_clearing}
\mu_{xy} = \boldsymbol \mu_{xy}^\alpha (\tau^\alpha) = \boldsymbol \mu_{xy}^\gamma (\tau^\gamma),
~ \forall x \in \X, y \in \Y.
\end{equation}

\item \textbf{One-Sided Money Burning:}
There is no market segment $xy$ where both passengers and taxis wait a positive amount of time; namely,
\begin{equation}  \label{eq:one_sided_money_burning}
\min \big\{ \tau_{xy}^\alpha, \tau_{xy}^\gamma \big\} = 0,
~ \forall x \in \X, y \in \Y.
\end{equation}

\end{enumerate}

\end{definition}


We remark that the equilibrium notion from Definition \ref{def:stable_matching_random} admits a closed-form expression in the logit case:

\begin{example}  \label{ex:matching_logit}
If the random utility terms $(\varepsilon_{iy})_y$ and $(\eta_{xj})_x$ follow i.i.d.\ Gumbel distributions, the choice probabilities defined in equations \eqref{eq:demand_passengers}--\eqref{eq:demand_taxis} are logit probabilities,
and therefore the market clearing condition \eqref{eq:market_clearing} corresponds to the following system of equations:
\begin{gather*}
\mu_{xy}
= \mu_{x0} \exp \big( \alpha_{xy} - \tau_{xy}^\alpha \big)
= \mu_{0y} \exp \big( \gamma_{xy} - \tau_{xy}^\gamma \big), \\
\mu_{x0} = \frac{n_x}{1 + \sum_y \exp \big( \alpha_{xy} - \tau_{xy}^\alpha \big)},
\qquad
\mu_{0y} = \frac{m_y}{1 + \sum_x \exp \big( \gamma_{xy} - \tau_{xy}^\gamma \big)},
\end{gather*}
hence $\tau_{xy}^\alpha = \alpha_{xy} - \log ( \mu_{xy} / \mu_{x0} )$
and $\tau_{xy}^\gamma = \gamma_{xy} - \log ( \mu_{xy} / \mu_{0y} )$.
Imposing condition \eqref{eq:one_sided_money_burning}, we obtain
\begin{equation*}
\min \big\{ \alpha_{xy} - \log (\mu_{xy}/\mu_{x0}), 
\gamma_{xy} - \log (\mu_{xy}/\mu_{0y}) \big\} = 0,
\end{equation*}
from which we get
\begin{equation*}
\mu_{xy} = \min \big\{ \mu_{x0} \exp (\alpha_{xy}),
\mu_{0y} \exp (\gamma_{xy}) \big\}.
\end{equation*}
Substituting this expression of $\mu_{xy}$ into the accounting equations $\mu_{x0} + \sum_{y \in \Y} \mu_{xy} = n_x$ and $\mu_{0y} + \sum_{x \in \X} \mu_{xy} = m_y$ yields the following system in $\mu_{x0}$ and $\mu_{0y}$:
\begin{equation}  \label{eq:logit_matching}
\begin{array}{l}
\mu_{x0} + \sum_{y \in \Y} \min \big\{ \mu_{x0} \exp (\alpha_{xy}), \mu_{0y} \exp (\gamma_{xy}) \big\} = n_x \\
\mu_{0y} + \sum_{x \in \X} \min \big\{ \mu_{x0} \exp (\alpha_{xy}), \mu_{0y} \exp (\gamma_{xy}) \big\} = m_y.
\end{array}
\end{equation}
\end{example}

\subsection{Existence and Uniqueness}

The existence of a unique solution to the system of equations \eqref{eq:logit_matching} can be established by applying a fixed-point theorem.
For general random taste shifters beyond the logit case, however, the choice probabilities do not admit a closed-form expression and we must rely on other mathematical methods.
First, we define $\tau_{xy} = \tau_{xy}^\alpha - \tau_{xy}^\gamma$.
Clearly, $\tau_{xy}^\alpha$ and $\tau_{xy}^\gamma$ can be treated as the positive and negative parts of $\tau_{xy}$:
\begin{equation}  \label{eq:defTaus}
\tau_{xy}^\alpha = \tau_{xy}^+ = \max \{\tau_{xy},0\},
\qquad
\tau_{xy}^\gamma = \tau_{xy}^- = -\min \{\tau_{xy},0\}.
\end{equation}

Notice that by definition of these positive and negative parts, condition \eqref{eq:one_sided_money_burning} is satisfied automatically.
With condition \eqref{eq:market_clearing}, we can thus characterize an aggregate stable matching with money burning as a solution to the system of nonlinear equations:
\begin{equation}  \label{eq:zero_excess_demand}
\mathbf e (\tau) = 0
\end{equation}
where $\mathbf e : \R^{|\X \times \Y|} \to \R^{|\X \times \Y|}$ is the excess demand function defined by
\begin{equation}  \label{eq:excess_demand}
\begin{aligned}
\mathbf e_{xy} (\tau)
&:= \boldsymbol \mu_{xy}^\gamma (\tau^-) - \boldsymbol \mu_{xy}^\alpha (\tau^+). 
\end{aligned}
\end{equation}
Using this formulation, we leverage results from \citeasnoun{Rheinboldt1974} on M-functions to prove both existence and uniqueness of the aggregate stable matching.%
\footnote{Uniqueness is driven by the fact that the distributions of the random utility components are continuous.
By contrast, in the case of deterministic
utilities as studied in section \ref{sec:matching_deterministic_utility}, there may
exist multiple equilibria.}
(See Appendix \ref{apx:Mfunctions} on M-functions, or \citeasnoun{GalichonJacquet2024} for an extended review.)

\begin{theorem}  \label{thm:ExistenceUniquenessEquil}
Under Assumption \ref{ass:nonvanishing}, there exists a unique aggregate stable matching with money burning in the random utility case.
\end{theorem}



\subsection{Limit when the Stochastic Utility Component is Small}
\label{par:limit}

In this paragraph, we show that the aggregate stable matching with the logit stochastic component studied in section \ref{sec:matching_random_utility} converges (when the amount of randomness tends to zero) to an aggregate stable matching with deterministic utility as studied in section \ref{sec:matching_deterministic_utility}.
To do this, consider a model where the stochastic utility components are logit with scaling parameter $\sigma > 0$.
Extending the analysis from Example \ref{ex:matching_logit}, the aggregate stable matching $\mu$ is given as a function of $\sigma$ by
\begin{equation}
\mu_{xy} (\sigma) = \min \big\{
\mu_{x0} (\sigma) e^{\alpha_{xy}/\sigma},
\mu_{0y} (\sigma) e^{\gamma_{xy}/\sigma}
\big\},
\end{equation}
where $\mu_{x0} (\sigma)$ and $\mu_{0y} (\sigma)$ solve the system
\begin{equation} \label{eq:system_mu0_sigma}
\begin{array}{l}
\mu_{x0} (\sigma) + \sum_y \min \big\{
\mu_{x0} (\sigma) e^{\alpha_{xy}/\sigma},
\mu_{0y} (\sigma) e^{\gamma_{xy}/\sigma}
\big\} = n_x, \\
\mu_{0y} (\sigma) + \sum_x \min \big\{
\mu_{x0} (\sigma) e^{\alpha_{xy}/\sigma},
\mu_{0y} (\sigma) e^{\gamma_{xy}/\sigma}
\big\} = m_y.
\end{array}
\end{equation}
Then, the following theorem holds:

\begin{theorem} \label{thm:limit}
There are vectors $(u_x) \in \R_+^\X$ and $(v_y) \in \R_+^\Y$ and a matrix $(\mu_{xy}) \in \R_+^{\X \times \Y}$ such that, up to subsequence extraction,
$u_x = -\lim_{\sigma \to 0} \sigma \ln \mu_{x0} (\sigma)$ and
$v_y = -\lim_{\sigma \to 0} \sigma \ln \mu_{0y}(\sigma)$,
and $(\mu,u,v)$ is an aggregate stable matching with money burning from Definition \ref{def:stable_matching_deterministic}.
\end{theorem}

An immediate implication of Theorem \ref{thm:limit} is the existence of an aggregate stable matching with money burning in the deterministic utility case, as studied in section \ref{sec:matching_deterministic_utility}.

%
%

\subsection{Stationary dynamic interpretation}
\label{sec:dynamics}

We further consider a (discrete-time) dynamic model in which the stationary equilibrium coincides with the equilibrium from Definition \ref{def:stable_matching_random} in a static model.
At each period, there are $n_x$ passengers of type $x \in \X$ and $m_y$ taxis of type $y \in \Y$ joining in the market.
Again, the prices are fixed.

The platform tries to clear the market insofar as possible;
however, queues must be formed since in general the number of type-$x$ passengers requesting a type-$y$ taxi at a given time does not coincide with the number of type-$y$ taxis opting to pick up a type-$x$ passenger.
We let $Q_{xy}^\alpha (t)$ be the number of passengers of type $x$ already queuing for a taxi of type $y$ at the beginning of period $t$, and $Q_{xy}^\gamma (t)$ be the number of taxis of type $y$ already queuing for a passenger of type $x$ at the beginning of period $t$.
These queues are the money-burning device that induces waiting times.
If there is no queue, there is zero waiting time:
$\tau_{xy}^\alpha (t) = 0$ if and only if $Q_{xy}^\alpha (t) = 0$,
and $\tau_{xy}^\gamma (t) = 0$ if and only if $Q_{xy}^\gamma (t) = 0$.%
\footnote{A full dynamic queueing model would require specifying a technology mapping queue lengths into waiting times.
For the stationary characterization here, however, that mapping is immaterial: the argument only uses the fact that positive waiting time corresponds to a nonempty queue.}

Focusing on passengers for now, utility is still the sum of a systematic term $\alpha_{xy}$, a random utility term $\varepsilon_{iy}$, and it is quasi-linear in the time waited.
We assume that agents are not forward looking: passengers base their choice upon the current waiting time for taxis of type $y$,
$\tau_{xy}^\alpha (t)$, yielding the decision utility $\alpha_{xy} - \tau_{xy}^\alpha (t) + \varepsilon_{iy}$ for a type-$y$ taxi.
Passengers can also opt out, in which case their systematic utility is normalized to zero.

For exposition purposes, we assume that the random utility components $\varepsilon_{iy}$ are i.i.d.\ logit as in Example \ref{ex:matching_logit}.
As a result, the proportion of type-$x$ passengers who opt for a taxi of type $y$ at time $t$ is
\begin{equation*}
\frac{ \exp \big( \alpha_{xy} - \tau_{xy}^\alpha (t) \big) }{
1 + \sum_{y' \in \Y} \exp \big( \alpha_{xy'} - \tau_{xy'}^\alpha (t) \big)
}.
\end{equation*}
Likewise, taxi $j$ of type $y$ enjoys $\gamma_{xy} - \tau_{xy}^\gamma + \eta_{xj}$ for picking up a type-$x$ passenger.
Taxis also get a systematic utility normalized to zero if they opt out.
Under the same logit assumption, the proportion of type-$y$ taxis who opt for a passenger of type $x$ at time $t$ is
\begin{equation*}
\frac{ \exp \big( \gamma_{xy} - \tau_{xy}^\gamma (t) \big) }{
1 + \sum_{x' \in \X} \exp \big( \gamma_{x'y} - \tau_{x'y}^\gamma (t) \big)
}.
\end{equation*}

We further assume that once their decision is made, agents stay in the same queue.
Thus, once the new wave of passengers has made their choice, the number of type-$x$ passengers lining up for type-$y$ taxis is
\begin{equation*}
Q_{xy}^\alpha(t) + \frac{n_x \exp \big( \alpha_{xy} - \tau_{xy}^\alpha (t) \big) }{
1 + \sum_{y' \in \Y} \exp \big( \alpha_{xy'} - \tau_{xy'}^\alpha (t) \big)
},
\end{equation*}
i.e., those who were already queuing from the previous period plus the newly arrived passengers incrementing the queue.
Similarly, there are
\begin{equation*}
Q_{xy}^\gamma(t) + \frac{m_y \exp \big( \gamma_{xy} - \tau_{xy}^\gamma (t) \big) }{
1 + \sum_{x' \in \X} \exp \big( \gamma_{x'y} - \tau_{x'y}^\gamma (t) \big)
}
\end{equation*}
type-$y$ taxis lining up for type-$x$ passengers, again arising from the queue at the previous period plus the newly arrived taxis.

Out of those two queues, the platform clears out a total number $\mu_{xy}(t)$ of $xy$ matches in period $t$ equal to the size of the shortest queue, that is
\begin{equation*}
\mu_{xy}(t) = \min \left\{
\begin{array}{c}
Q_{xy}^\alpha(t) + \frac{n_x \exp ( \alpha_{xy} - \tau_{xy}^\alpha (t) ) }{
1 + \sum_{y' \in \Y} \exp ( \alpha_{xy'} - \tau_{xy'}^\alpha (t) )
}, \\
Q_{xy}^\gamma(t) + \frac{m_y \exp ( \gamma_{xy} - \tau_{xy}^\gamma (t) ) }{
1 + \sum_{x' \in \X} \exp ( \gamma_{x'y} - \tau_{x'y}^\gamma (t) )
}
\end{array}
\right\},
\end{equation*}
and the lengths of the queues are therefore updated for the next period as
\begin{equation*}
\left\{
\begin{array}{l}
Q_{xy}^\alpha (t+1) =
Q_{xy}^\alpha(t) + \frac{n_x \exp ( \alpha_{xy} - \tau_{xy}^\alpha (t) ) }{
1 + \sum_{y' \in \Y} \exp ( \alpha_{xy'} - \tau_{xy'}^\alpha (t) )
} - \mu_{xy} (t)
\\
Q_{xy}^\gamma (t+1) =
Q_{xy}^\gamma(t) + \frac{m_y \exp ( \gamma_{xy} - \tau_{xy}^\gamma (t) ) }{
1 + \sum_{x' \in \X} \exp ( \gamma_{x'y} - \tau_{x'y}^\gamma (t) )
} - \mu_{xy} (t).
\end{array}
\right.
\end{equation*}
Clearly, $\min \big\{ Q_{xy}^\alpha (t+1), Q_{xy}^\gamma (t+1) \big\} = 0$; therefore
\begin{equation*}
\min \big\{ \tau_{xy}^\alpha (t+1), \tau_{xy}^\gamma (t+1) \big\} = 0.
\end{equation*}

In the stationary state, the lengths of the queues and the waiting times remain constant.
As a result,
\begin{equation*}
\mu_{xy} = \frac{n_x \exp \big( \alpha_{xy}-\tau_{xy}^\alpha\big) }{
1 + \sum_{y' \in \Y} \exp \big( \alpha_{xy'} - \tau_{xy'}^\alpha \big) }
= \frac{m_y \exp \big( \gamma_{xy} - \tau_{xy}^\gamma \big) }{
1 + \sum_{x' \in \X} \exp \big( \gamma_{x'y} - \tau_{x'y}^\gamma \big) }
\end{equation*}
and
\begin{equation*}
\min \big\{ \tau_{xy}^\alpha, \tau_{xy}^\gamma \big\} = 0,
\end{equation*}
which are exactly the conditions for an aggregate stable matching with money burning for the static model (Definition \ref{def:stable_matching_random}).
It is straightforward to extend this analysis beyond the logit case to more general distributions, by replacing the logit choice probabilities with the general choice probabilities featured in equation \eqref{eq:market_clearing}.


\section{Deferred Acceptance for Matching with Random Utility}
\label{sec:deferred_acceptance}

In this section, we propose a deferred acceptance algorithm for our matching model with a continuum of agents, observable types, idiosyncratic shocks, and waiting lines.
The key insight is to view \citeasnoun{GaleShapley1962}'s classical algorithm as iterating between two \emph{constrained discrete-choice} problems.
At each round, one side chooses under an endogenous availability cap determined by the offers received from the other side.
In our setting, these caps are enforced by waiting times, which act as shadow prices: over-requested segments carry positive waiting times that reduce demand to capacity, while under-requested segments have zero waiting time.
The algorithm alternates the induced constrained-demand maps on both sides until proposals and tentative acceptances coincide, yielding an aggregate stable matching.

\subsection{Constrained choice.}
\label{sec:constrained_choice}
We begin by formalizing the constrained choice problem faced by a single side.
Recall from section \ref{sec:matching_random_utility} that
$\boldsymbol \mu_{xy}^\alpha (\tau^\alpha)$ denotes the the passenger-side demands, and
$\boldsymbol \mu_{xy}^\gamma (\tau^\gamma)$ the taxi-side demands, as functions of their waiting times.
Since the same construction will apply to both sides (passengers choosing taxis, or taxis choosing passengers), we lighten notation by dropping superscripts in this subsection.
Thus, $\boldsymbol \mu(\tau)$ denotes the demand induced by a waiting-time vector $\tau$, where $\boldsymbol\mu$ and $\tau$ can stand either for $\boldsymbol\mu^\alpha$ and $\tau^\alpha$, or for $\boldsymbol\mu^\gamma$ and $\tau^\gamma$.

We fix a capacity matrix $\bar \mu = (\bar \mu_{xy})$ and think of $\bar\mu_{xy} > 0$ as the maximum number of matches of type $xy$ that can be accommodated on this side of the market.%
\footnote{Capacities must be strictly positive, as Assumption \ref{ass:nonvanishing} excludes zero demand regardless of waiting time.}
Waiting times act as the instrument that enforces these constraints:
if a segment $xy$ is over-demanded, the waiting time $\tau_{xy}$ should be positive to adjust demand to capacity;
but if the segment is under-demanded, the waiting time should be zero.

Formally, we look for $\tau \geq 0$ such that (i) demand does not exceed capacity in any segment, and (ii) waiting lines form only when the capacity constraint is saturated:
\begin{equation}  \label{eq:cap_ineq}
\boldsymbol \mu_{xy} (\tau) \leq \bar \mu_{xy},
\quad 
\tau_{xy} \big( \bar \mu_{xy} - \boldsymbol \mu_{xy} (\tau) \big) =0,
\qquad \forall x \in \X, y \in \Y.
\end{equation}
If the capacity constraint is not binding ($\boldsymbol \mu_{xy}(\tau) < \bar \mu_{xy}$), then $\tau_{xy}=0$;
conversely, if $\tau_{xy} > 0$, then the constraint must bind ($\boldsymbol \mu_{xy}(\tau) = \bar \mu_{xy}$).
As in section \ref{sec:matching_random_utility}, it is convenient to encode this complementarity structure through a single unconstrained variable.
To this end, introduce the \emph{overcapacity} variable $\rho_{xy} \geq 0$ as the slack of the capacity constraint:
\begin{equation*}
\rho_{xy} = \bar \mu_{xy} - \boldsymbol \mu_{xy} (\tau).
\end{equation*}
We can then consider $\tau$ and $\rho$ as the positive and negative parts of an unconstrained variable $\theta = \tau - \rho$, so that
\begin{equation*}
\tau = \theta^+
\quad \text{and} \quad
\rho = \theta^-.
\end{equation*}
The problem \eqref{eq:cap_ineq} then boils down to finding $\theta$ such that
\begin{equation} \label{eq:theta_equation}
\boldsymbol\mu(\theta^+) + \theta^- = \bar\mu.
\end{equation}

As for Theorem \ref{thm:ExistenceUniquenessEquil}, results on M-functions from \citeasnoun{Rheinboldt1974} (see Appendix \ref{apx:Mfunctions}) allow us to prove existence and uniqueness of the solution to this problem.

\begin{theorem} \label{thm:constrained_choice_unique}
Under Assumption \ref{ass:nonvanishing}, there exists a unique solution $\theta$ to the constrained choice problem \eqref{eq:theta_equation}.
\end{theorem}

Given $\bar \mu$, we define the resulting \emph{constrained demand} on this side of the market by
\begin{equation*}
\boldsymbol c_{xy} (\bar \mu) = \boldsymbol \mu_{xy} (\theta^+),
\qquad \forall x \in \X, y \in \Y,
\end{equation*}
where $\theta$ is the unique solution to equation \eqref{eq:theta_equation}.
Reintroducing the superscript notations, $\boldsymbol c^\alpha (\bar \mu^\alpha)$ is thus the constrained demand by passengers $x$ for taxis $y$ under the capacity constraints $\bar \mu^\alpha$,
and $\boldsymbol c^\gamma (\bar \mu^\gamma)$ is the constrained demand by taxis $y$ for passengers $x$ under the capacity constraints $\bar \mu^\gamma$.

\begin{example}
In the logit case, the problem \eqref{eq:cap_ineq} is to find $\tau \geq 0$ and $\mu \geq 0$ such that
\begin{equation*}
\mu_{xy} = \mu_{x0 }\exp (\alpha_{xy} - \tau_{xy}),
\qquad
\mu_{xy} \leq \bar \mu_{xy},
\qquad
\min \big\{ \tau_{xy}, \bar \mu_{xy} - \mu_{xy} \big\} = 0.
\end{equation*}
From these conditions we obtain
\begin{equation*}
\mu_{xy} = \min \big\{ \mu_{x0} \exp (\alpha_{xy}), \bar \mu_{xy} \big\}
\end{equation*}
and $\mu_{x0}$ is therefore solution to the scalar equation
\begin{equation}
\mu_{x0} + \sum_{y \in \Y} \min \big\{ \mu_{x0} \exp (\alpha_{xy}), \bar \mu_{xy} \big\} = n_x.
\end{equation}
This equation has a unique solution since the left-hand side is continuous and strictly increasing in $\mu_{x0}$ from $\R_+$ to $\R_+$.
\end{example}

\subsection{Deferred Acceptance.}
Recall the principle of the classical Gale--Shapley deferred acceptance algorithm:
passengers make offers to taxis; taxis tentatively keep their favorite offers and reject the rest.
In the next round, rejected passengers make offers to taxis who have not yet rejected them.
This process repeats until no rejection occurs, and the resulting matching is stable.

We adapt this idea to our setting with a continuum of agents, observable types, idiosyncratic shocks, and waiting lines.
The algorithm alternates between two constrained choice problems as studied in section \ref{sec:constrained_choice}, and keeps track of which offers remain available over time.
Let $\mu_{xy}^{A,t-1}$ denote the number of offers that are still available from passengers of type $x$ to taxis of type $y$ at the beginning of round $t$; initially, all offers are available up to the binding population constraint, so $\mu_{xy}^{A,0}=\min\{n_x,m_y\}$.
Given available offers $\mu^{A,t-1}$, passengers form their constrained demand and generate a proposal matrix $\mu^{P,t} = \boldsymbol c^\alpha (\mu^{A,t-1})$, so that $\mu_{xy}^{P,t}$ is the volume of offers from passengers $x$ to taxis $y$ at round $t$.
Taxis then solve the analogous constrained choice problem given the incoming proposals, and keep (tentatively accept) a matrix $\mu^{K,t} = \boldsymbol c^\gamma (\mu^{P,t})$, so that $\mu_{xy}^{K,t}$ is the volume of offers from passengers $x$ to taxis $y$ that taxis keep at round $t$.
Rejections equal $\mu^{P,t}-\mu^{K,t}$ and are removed from future availability, so the available offer pool shrinks over time until no rejection remains.
Formally, the algorithm is described as follows:

\begin{algorithm}  \label{alg:darum}
Step $0$.
Initialize the number of offers available to passengers as
\begin{equation*}
\mu_{xy}^{A,0} = \min \{n_x, m_y\}.
\end{equation*}

Step $t \geq 1$. There are three phases:

\underline{Proposal phase}: Passengers propose offers subject to
availability constraints:
\begin{equation*}
\mu^{P,t} = \boldsymbol c^\alpha (\mu^{A,t-1}).
\end{equation*}

\underline{Disposal phase}: Taxis keep their best offers among the
proposals:
\begin{equation*}
\mu^{K,t} = \boldsymbol c^\gamma (\mu^{P,t}).
\end{equation*}

\underline{Update phase}:
The offers rejected by taxis are removed from the pool of offers available to passengers:
\begin{equation*}
\mu^{A,t} = \mu^{A,t-1} - \big( \mu^{P,t} - \mu^{K,t} \big).
\end{equation*}

The algorithm stops and returns $\mu = \mu^{P,t}$ when the norm of $\mu^{P,t} - \mu^{K,t}$ is below some
tolerance level.
\end{algorithm}

Numerically, the proposal and disposal phases each require solving the constrained choice problem \eqref{eq:theta_equation}.
This can be achieved using the constructive method found in the proof of Theorem \ref{thm:constrained_choice_unique}.
The next theorem establishes that this procedure converges and characterizes its limit as the unique aggregate stable matching with money burning.


\begin{theorem}  \label{thm:DA_convergence}
Under Assumption \ref{ass:nonvanishing}, Algorithm \ref{alg:darum} converges to a limit matching $\mu$.
Furthermore, let $\theta^\alpha$ be the unique solution to \eqref{eq:theta_equation} with
$\boldsymbol\mu=\boldsymbol\mu^\alpha$ and $\bar\mu=\mu$, and define $\tau^\alpha=(\theta^\alpha)^+$.
Likewise, let $\theta^\gamma$ be the unique solution to \eqref{eq:theta_equation} with
$\boldsymbol\mu=\boldsymbol\mu^\gamma$ and $\bar\mu=\mu$, and define $\tau^\gamma=(\theta^\gamma)^+$.
Then $(\mu,\tau^\alpha,\tau^\gamma)$ is the aggregate stable matching with money burning for the random utility model.
\end{theorem}

A consequence of Theorem \ref{thm:DA_convergence} is that, unlike in classical Gale--Shapley deferred acceptance for the finite-agent model, the output of Algorithm \ref{alg:darum} does not depend on which side proposes first.
Indeed, uniqueness of the aggregate stable matching with money burning (Theorem \ref{thm:ExistenceUniquenessEquil}) implies that the same stable matching is reached regardless of which side makes offers and which side accepts them.
Moreover, Algorithm \ref{alg:darum} converges only asymptotically, rather than reaching an exact outcome in finitely many steps as in Gale--Shapley deferred acceptance.
These differences reflect the fact that our framework is designed for a large-market environment with aggregate uncertainty and type heterogeneity, rather than for a finite set of agents with ordinal preferences.

The iterative proposal--retention procedure of Algorithm \ref{alg:darum} belongs to the broad class of deferred acceptance algorithms initiated by \citeasnoun{GaleShapley1962} and extended by \citeasnoun{KelsoCrawford1982} and \citeasnoun{HatfieldMilgrom2005}.
But while these papers study finite-agent markets, typically many-to-one and with salaries or contracts, our setting is a one-to-one, fixed-price NTU environment with a continuum of agents aggregated into observable types.
Our framework also imposes free disposal and delivers an aggregate equilibrium with equal treatment within type, which may not be the case for the stable outcomes reached by the centralized matching procedures mentioned above.
Crucially, in our framework overdemand is resolved not through contract terms but through endogenous money burning, leading to an equilibrium where all agents obtain their first choice.


\section{Conclusion}
\label{sec:conclusion}


We conclude by discussing the potential econometric consequences of our equilibrium notion relative to existing models.
A key implication of our equilibrium notion is that, in the logit case, it delivers a Leontief aggregate matching function which contrasts with two prominent alternatives. In the Dagsvik--Menzel framework, where idiosyncratic tastes vary at the individual-identity level, the implied NTU matching function takes a multiplicative form with scale effects,
\begin{equation}
\mu_{xy} = \mu_{x0} \mu_{0y} \exp (\alpha_{xy}+\gamma_{xy}),
\label{NTU_MMF_Menzel}
\end{equation}
while in Choo and Siow's separable TU logit model the matching function is Cobb--Douglas,%
\footnote{See \citeasnoun{MourifieSiow2021} for a survey of the aggregate matching function.}
\begin{equation}
\mu_{xy} = \sqrt{\mu_{x0}\mu_{0y}\exp (\alpha_{xy}+\gamma_{xy})}.
\label{Choo-SiowMMF}
\end{equation}
By contrast, our notion of aggregate stable matching with one-sided money burning yields the Leontief form
\begin{equation}
\mu_{xy} = \min \big\{ \mu_{x0}\exp \alpha_{xy},\mu_{0y}\exp \gamma_{xy} \big\}.
\label{ourMMF}
\end{equation}
This difference reflects the microeconomic mechanism that clears the market. With fixed prices and no transfers, scarcity is resolved by waiting on the short side, so realized matches are determined by the binding side of the market-segment constraint.
This Leontief functional form for the matching function has distinctive econometric implications, which we leave open for future work.



\bibliographystyle{econometrica}
\bibliography{refAll}

@Article{HatfieldMilgrom2005,
  author        = {Hatfield, John William and Milgrom, Paul R},
  title         = {Matching with contracts},
  journal       = {American Economic Review},
  year          = {2005},
  pages         = {913--935},
  __markedentry = {[Baiyu:]},
}

@Article{ChooSiow2006,
  author  = {Choo, E. and A. Siow},
  title   = {Who Marries Whom and Why},
  journal = {Journal of Political Economy},
  year    = {2006},
  volume  = {114},
  number  = {1},
  pages   = {175-201},
}

@Article{Dagsvik2000,
  author  = {Dagsvik, J. K.},
  title   = {Aggregation in Matching Markets},
  journal = {International Economic Review},
  year    = {2000},
  volume  = {41},
  number  = {1},
  pages   = {27-57},
}

@Article{GalichonSalanie2022,
  author  = {Galichon, A. and B. Salani\'{e}},
  title   = {Cupid’s Invisible Hand: Social Surplus and Identification in Matching Models},
  journal = {The Review of Economic Studies},
  year    = {2022},
  volume  = {89},
  number  = {5},
  pages   = {2600--2629},
}

@Article{Menzel2015,
  author  = {Menzel, K.},
  title   = {Matching Markets as Two-Sided Demand Systems},
  journal = {Econometrica},
  year    = {2015},
  volume  = {83},
  number  = {3},
  pages   = {897-941},
}

@Article{ShapleyShubik1972,
  author  = {Shapley, L. S. and M. Shubik},
  title   = {The Assignment Game, I: The Core},
  journal = {International Journal of Game Theory},
  year    = {1972},
  volume  = {1},
  pages   = {111-130},
}

@Article{Becker1973,
  author  = {G. S. Becker},
  title   = {A theory of marriage: part I},
  journal = {Journal of Political Economy},
  year    = {1973},
  volume  = {81},
  pages   = {813-846},
}

@Article{CheKoh2016,
  author  = {Che, Y.-K. and Koh, Y.},
  title   = {Decentralized College Admissions},
  journal = {Journal of Political Economy},
  year    = {2016},
  volume  = {124},
  pages   = {1295-1337},
}

@Article{EcheniqueLeeShumYenmex,
  author  = {Echenique, F. and S. M. Lee and M. Shum and M. B. Yenmez},
  title   = {The Revealed Preference Theory of Stable and Extremal Stable Matchings},
  journal = {Econometrica},
  year    = {2013},
  volume  = {81},
  pages   = {153-171},
}

@Article{GaleShapley1962,
  author  = {Gale, D. and L. S. Shapley},
  title   = {College Admissions and the Stability of Marriage},
  journal = {The American Mathematical Monthly},
  year    = {1962},
  volume  = {69},
  pages   = {9-15},
}

@book{McFadden1976,
  title={The mathematical theory of demand models},
  author={MacFadden, Daniel L},
  year={1976},
  publisher={Lexington Books}
}

@Article{GalichonKominersWeber2019,
  author  = {Galichon, A. and S. Kominers and S. Weber},
  title   = {Costly Concessions: An Empirical Framework for Matching with Imperfectly Transferable
  Utility},
  journal = {Journal of Political Economy},
  year    = {2019},
  volume  = {127},
  number  = {6},
  pages   = {2875-2925},
}

@Article{MourifieSiow2021,
  author  = {Mourifi\'{e}, I. and A. Siow},
  title   = {The Cobb-Douglas marriage matching function: Marriage matching with peer and scale effects},
  journal = {Journal of Labor Economics},
  year    = {2021},
  volume  = {39},
  number  = {S1},
  pages   = {S239--S274},
}

@Article{Bonnet2020,
  author  = {O. Bonnet and A. Galichon and Y.-W. Hsieh and K. O'Hara and M. Shum},
  title   = {Yogurts Choose Consumers? Estimation of Random-Utility Models via Two-Sided Matching},
  journal = {The Review of Economic Studies},
  year    = {2022},
  volume  = {89},
  number  = {6},
  pages   = {3085--3114},
}

@Article{EcheniqueYariv2013,
  author  = {Echenique, F. and Alejandro Robinson--Cortés and L. Yariv},
  title   = {An Experimental Study of Decentralized Matching},
  journal = {Quantitative Economics},
  year    = {2025},
  volume  = {16},
  pages   = {497–533},
}

@Article{Gale1996,
  author  = {D. Gale},
  title   = {Equilibria and Pareto Optima of Markets with Adverse Selection},
  journal = {Economic Theory},
  year    = {1996},
  volume  = {7},
  pages   = {207-235},
}

@Article{NiederleYariv2009,
  author  = {Niederle, M. and L. Yariv},
  title   = {Decentralized Matching with Aligned Preferences},
  journal = {NBER Working Paper},
  year    = {2009},
}

@Article{GalichonHsieh_hedonic,
  author  = {A. Galichon and Hsieh, Y.-W.},
  title   = {A Hedonic Model with Rationing by Waiting},
  journal = {Working paper},
  year    = {2020},
}

@Article{RothRothblumVate1993,
  author  = {Roth, A. and Rothblum, U. and Vande Vate, J.},
  title   = {Stable Matchings, Optimal Assignments, And Linear Programming},
  journal = {Mathematics of Operations Research},
  year    = {1993},
  volume  = {18},
  pages   = {803-828},
}

@Book{Rheinboldt1974,
  title     = {Methods For Solving Systems Of Nonlinear Equations},
  publisher = {SIAM},
  year      = {1974},
  author    = {Rheinboldt, W.},
}

@Book{Sandel2013,
  title     = {What Money Can't Buy: The Moral Limits of Markets},
  publisher = {Farrar, Straus and Giroux},
  year      = {2013},
  author    = {Sandel, M.},
}

@Article{Smith2006,
  author  = {Smith, L.},
  title   = {The Marriage Model With Search Frictions},
  journal = {Journal of Political Economy},
  year    = {2006},
  volume  = {114},
  pages   = {1124-1144},
}

@Article{KelsoCrawford1982,
  author        = {Kelso, Alexander S and Crawford, Vincent P},
  title         = {Job matching, coalition formation, and gross substitutes},
  journal       = {Econometrica},
  year          = {1982},
  pages         = {1483--1504},
  __markedentry = {[Baiyu:]},
}

@Article{Iversen1993,
  author  = {Iversen, T.},
  title   = {A Theory of Hospital Waiting List},
  journal = {Journal of Health Economics},
  year    = {1993},
  volume  = {12},
  pages   = {55-71},
}

@Article{LindsayFeigenbaum1984,
  author  = {Lindsay, C. M. and B. Feigenbaum},
  title   = {Rationing by Waiting Lists},
  journal = {American Economic Review},
  year    = {1984},
  volume  = {74},
  number  = {3},
  pages   = {404-417},
}

@InBook{IversenSiciliani2011,
  chapter   = {Non-Price Rationing and Waiting Times},
  pages     = {649-670},
  title     = {Oxford Handbook of Health Economics},
  publisher = {Oxford University Press},
  year      = {2011},
  author    = {Iversen, T. and L. Siciliani},
}

@Article{MartinSmith1999,
  author  = {Martin, S. and P. C. Smith},
  title   = {Rationing by Waiting Lists: An Empirical Investigation},
  journal = {Journal of Public Economics},
  year    = {1999},
  volume  = {71},
  pages   = {141-164},
}

@InBook{Maddala1986,
  chapter   = {Disequilibrium, Self-Section, and Switching Models},
  pages     = {1632-1688},
  title     = {Handbook of Econometrics, vol III},
  publisher = {Elsevier},
  year      = {1986},
  author    = {Maddala, G. S.},
}

@Article{Braverman_etal_2016,
  author  = {Braverman, M. and J. Chen and S. Kannan},
  title   = {Optimal Provision-After-Wait in Healthcare},
  journal = {Mathematics of Operations Research},
  year    = {2016},
  volume  = {41},
  pages   = {352-376},
}

@Article{Benassy1976,
  author  = {B\'{e}nassy, J.-P.},
  title   = {The Disequilibrium Approach to Monopolistic Price Setting and General Monopolistic
  Equilibrium,},
  journal = {Review of Economic Studies},
  year    = {1976},
  volume  = {43},
  pages   = {69-81},
}

@Article{BurdettColes1997,
  author  = {Burdett, K. and Coles, M.},
  title   = {Marriage and Class},
  journal = {Quarterly Journal of Economics},
  year    = {1997},
  volume  = {112},
  pages   = {252-168},
}

@Article{Dreze1987,
  author  = {J. Dr\'{e}ze},
  title   = {Underemployment Equilibria: From Theory to Econometrics and Policy},
  journal = {European Economic Review},
  year    = {1987},
  volume  = {31},
  pages   = {9-34},
}

@Article{FairJaffee1972,
  author  = {Fair, R. C. and D. M. Jaffee},
  title   = {Methods of Estimation for Markets in Disequilibrium},
  journal = {Econometrica},
  year    = {1972},
  volume  = {40},
  pages   = {497-514},
}

@Article{GlaeserLuttmer2003,
  author  = {Glaeser, E. and Luttmer, E.},
  title   = {The Misallocation Of Housing Under Rent Control},
  journal = {American Economic Review},
  year    = {2003},
  volume  = {93},
  number  = {4},
  pages   = {1027-1046},
}

@Article{GourierouxLaffontMonfort1980,
  author  = {Gourieroux, C. and J. J. Laffont and A. Monfort},
  title   = {Disequilibrium Econometrics in Simultaneous Equations Systems},
  journal = {Econometrica},
  year    = {1980},
  volume  = {48},
  pages   = {75-96},
}

@Article{GourierouxLaroque1985,
  author  = {Gourieroux, C. and Laroque, G},
  title   = {The Aggregation of Commodities in Quantity Rationing Models},
  journal = {International Economic Review},
  year    = {1985},
  volume  = {26},
  pages   = {681-699},
}

@Article{HartlineRoughgarden2008,
  author  = {Hartline, J. and T. Roughgarden},
  title   = {Mechanism Design and Money Burning},
  journal = {STOC},
  year    = {2008},
}

@Book{HassinHaviv2003,
  title     = {To Queue or Not to Queue: Equilibrium Behavior in Queuing Systems},
  publisher = {Kluwer Academic Publishers},
  year      = {2003},
  author    = {Hassin, R. and Haviv, M.},
}

@Book{Galichon2026,
  title     = {Discrete Choice Models: Mathematical Methods, Econometrics, and Data Science},
  publisher = {Princeton University Press},
  year      = {2026},
  author    = {Alfred Galichon},
}

@Article{Condorelli2012,
  author  = {Daniele Condorelli},
  title   = {What Money Can’t Buy: Efficient Mechanism Design With Costly Signals},
  journal = {Games and Economic Behavior},
  year    = {2012},
  volume  = {75},
  pages   = {613-624},
}

@Article{Sealy1979,
  author  = {Sealy, C. W. Jr.},
  title   = {Credit Rationing in the Commercial Loan Market: Estimates of a Structural Model Under
  Conditions of Disequilibrium},
  journal = {Journal of Finance},
  year    = {1979},
  volume  = {34},
  pages   = {689-702},
}

@Article{AzevedoLeshno2016,
  author  = {Azevedo, E. and Leshno, J.},
  title   = {A Supply and Demand Framework for Two-Sided Matching Markets},
  journal = {Journal of Political Economy},
  year    = {2016},
  volume  = {124},
  pages   = {1235-1268},
}

@misc{GalichonJacquet2024,
  author    = {Alfred Galichon and Antoine Jacquet},
  title     = {Substitutability, equilibrium transport, and matching models},
  editor    = {Young-Heon Kim and Soumik Pal and Brendan Pass},
  year      = {2024},
  note      = {Lecture notes prepared for the 2022 PIMS-IFDS-NSF summer school on optimal transport, University of Washington, Seattle (forthcoming in a volume edited by Y.-H.\ Kim, S.\ Pal, and B.\ Pass)}
}

@article{GalichonHsiehSylvestre2024,
  title = {Monotone Comparative Statics for Submodular Functions, with an Application to Deferred Acceptance},
  journal = {Working paper},
  author = {Alfred Galichon and Yu-Wei Hsieh and Maxime Sylvestre},
  year = {2024},
}

\appendix

\section{M-Functions}
\label{apx:Mfunctions}

Several results in our analysis rely on a theorem by \citeasnoun{Rheinboldt1974} regarding M-functions.%
\footnote{Our presentation differs slightly from that of \citeasnoun{Rheinboldt1974}; notably, strongly nonreversing functions are instead called P-functions, and his theorem is more general that the statement of Theorem \ref{thm:Mfunction_inverse_isotone}.}
%
(See also \citeasnoun{GalichonJacquet2024} for an extended presentation.)
Let $P$ be a subset of $\R^n$, and $f : P \to \R^n$ which maps $p = (p_i) \in P$ to $f(p) = (f_i(p)) \in \R^n$.

\begin{definition}  \label{def:Mfunction}
The function $f$ is an M-function if it is both: \begin{enumerate}[label=(\roman*), topsep= 0pt, itemsep= 0pt]
\item off-diagonally antitone:
for any $i \neq i'$, $f_i$ is weakly decreasing in $p_{i'}$,
\item strongly nonreversing: for any $p, p' \in P$, $p \leq p'$ and $f(p) \geq f(p')$ together imply $p = p'$.
\end{enumerate}
\end{definition}

\begin{theorem}[\citename{Rheinboldt1974} \citeyear*{Rheinboldt1974}, Theorem 9.1]
\label{thm:Mfunction_inverse_isotone}
If $f$ is an M-function, then it is \emph{inverse isotone}: for any $p, p' \in P$,
\[
f(p) \leq f(p') \implies p \leq p'.
\]
In particular, $f$ is injective.
\end{theorem}

We now apply Theorem \ref{thm:Mfunction_inverse_isotone} to the demand functions $\boldsymbol \mu^\alpha$ and $\boldsymbol \mu^\gamma$.

\begin{lemma}  \label{lem:demand_inverse_isotone}
Under Assumption \ref{ass:nonvanishing}, the functions $-\boldsymbol \mu^\alpha$ and $-\boldsymbol \mu^\gamma$ are inverse isotone.
\end{lemma}

\begin{proof}[Proof of Lemma \ref{lem:demand_inverse_isotone}]
We show that $-\boldsymbol \mu^\alpha$ is an M-function; the result then follows from Theorem \ref{thm:Mfunction_inverse_isotone}.
(The proof for $-\boldsymbol \mu^\gamma$ is similar.)
It is clear from expression \eqref{eq:demand_passengers} that $\boldsymbol \mu_{xy}^\alpha$ is weakly increasing in $\tau_{x'y'}^\alpha$ for $x'y' \neq xy$, hence $-\boldsymbol \mu^\alpha$ is off-diagonally antitone.
Next, let $\tau, \tau' \geq 0$ such that $\tau \leq \tau'$ and $-\boldsymbol \mu^\alpha (\tau) \geq -\boldsymbol \mu^\alpha (\tau')$ (we drop the superscript $\alpha$ from $\tau^\alpha$ for convenience).
Looking for a contradiction, assume $\tau \neq \tau'$.
By summation of $\boldsymbol \mu^\alpha (\tau) \leq \boldsymbol \mu^\alpha (\tau')$ we obtain
\[
\sum_x \big( n_x - \boldsymbol \mu_{x0}^\alpha (\tau) \big)
= \sum_{xy} \boldsymbol \mu_{xy}^\alpha (\tau)
\leq \sum_{xy} \boldsymbol \mu_{xy}^\alpha (\tau')
= \sum_x \big( n_x - \boldsymbol \mu_{x0}^\alpha (\tau') \big).
\]
Furthermore, under Assumption \ref{ass:nonvanishing}, for all $x$ and $y$ the function $\boldsymbol \mu_{x0}^\alpha$ is strictly increasing in $\tau_{xy}$, and constant in $\tau_{x'y}$ for all $x' \neq x$.
As a result, $\sum_x \big( n_x - \boldsymbol \mu_{x0}^\alpha (\tau) \big)$ is strictly decreasing in $\tau_{xy}$ for all $xy$.
Since $\tau \leq \tau'$ and $\tau \neq \tau'$ this implies
\[
\sum_x \big( n_x - \boldsymbol \mu_{x0}^\alpha (\tau) \big)
> \sum_x \big( n_x - \boldsymbol \mu_{x0}^\alpha (\tau') \big),
\]
a contradiction.
\end{proof}

\medskip

\section{Proofs}
\label{apx:proofs}

\subsection{Proof of Theorem \ref{thm:EquivWithClassicalNotion}} ~

\noindent \underline{Part (i)}.
Assume that $\pi$ is a stable matching in the classical sense (Definition \ref{def:classical_stable_matching}) and define $\mu$ as in \eqref{eq:mu_from_pi} and $u$ and $v$ as in \eqref{eq:u_v_from_pi}.
Clearly, one has that $\mu_{xy}$ is an integer, $\sum_{y \in \Y} \mu_{xy} \leq n_x$, and $\sum_{x \in \X} \mu_{xy} \leq m_y$.
Next, since $\alpha_{ij} = \alpha_{xy}$ and $\gamma_{ij} = \gamma_{xy}$, we have
\begin{equation*}
\min_{\substack{i: x_i= x \\ j: y_j= y}}
\max \big\{
u_i^\pi - \alpha_{ij},
v_j^\pi - \gamma_{ij}
\big\}
= \max \big\{
u_x - \alpha_{xy},
v_y - \gamma_{xy}
\big\}
\end{equation*}
as the two terms are separable in $i$ and $j$, so the minimum is attained by choosing the worst-off passenger of type $x$ and the worst-off taxi of type $y$.
Since
$\max \big\{
u_i^\pi - \alpha_{ij},
v_j^\pi - \gamma_{ij}
\big\} \geq 0$
for all $ij$ by Definition \ref{def:classical_stable_matching}(iv), we thus have
\begin{equation*}
\max \big\{
u_x - \alpha_{xy},
v_y - \gamma_{xy}
\big\} \geq 0.
\end{equation*}
Now assume $\mu_{xy} > 0$.
Then there are $i$ and $j$ such that $x_i=x$, $y_j=y$, and $\pi_{ij} > 0$,
so $u_i^\pi = \alpha_{ij}$ and $v_j^\pi = \gamma_{ij}$, hence
$\max \big\{
u_i^\pi - \alpha_{ij},
v_j^\pi - \gamma_{ij}
\big\} = 0$.
Thus from the equality above,
\begin{equation*}
\max \big\{
u_x - \alpha_{xy},
v_y - \gamma_{xy}
\big\} = 0.
\end{equation*}
Finally, $u_i^\pi \geq 0$ for all $i$ implies $u_x \geq 0$ by construction \eqref{eq:u_v_from_pi}, and assuming $\mu_{x0} > 0$, there must be $i$ such that $x_i = x$ and $\pi_{ij} = 0$ for all $j$, hence $u_i^\pi = 0$ from \eqref{eq:u_i_v_j} and thus $u_x = 0$ from \eqref{eq:u_v_from_pi}.
An analogous argument applies to $v_y$.

\medskip

\noindent \underline{Part (ii)}.
Assume that $(\mu ,u,v)$ is an aggregate stable matching with money burning (Definition \ref{def:stable_matching_deterministic}), and consider a matrix $(\pi_{ij}) \in \{0,1\}^{\I \times \J}$ such that $\sum_{j \in \J} \pi_{ij} \leq 1$, $\sum_{i \in \I} \pi_{ij} \leq 1$, and \eqref{eq:mu_from_pi} holds.
Looking for a contradiction, assume there is a blocking pair $ij$, so that
\begin{equation*}
\max \big\{
u_i^\pi - \alpha_{xy},
v_j^\pi - \gamma_{xy}
\big\} < 0
\end{equation*}
where $x$ and $y$ are the respective types of $i$ and $j$.
Suppose $i$ is unmatched, i.e., $\pi_{ij} = 0$ for all $j$.
Then $u_i^\pi = 0$ from \eqref{eq:u_i_v_j}, and we also have $\mu_{x0} > 0$, which implies $u_x = 0$ by Definition \ref{def:stable_matching_deterministic}(v); thus $u_x \leq u_i^\pi$.
Now suppose $i$ is matched, i.e., $\pi_{ij'} = 1$ for some $j'$.
Then $u_i^\pi = \alpha_{xy'}$ from \eqref{eq:u_i_v_j}, and denoting $y'$ the type of $j'$, we also have $\mu_{xy'} > 0$ hence
$\max \big\{
u_x - \alpha_{xy'},
v_{y'} - \gamma_{xy'}
\big\} = 0$ by Definition \ref{def:stable_matching_deterministic}(iv);
thus $u_x \leq \alpha_{xy'} = u_i^\pi$ again.
A similar argument shows that $v_y \leq v_j^\pi$ whether $j$ is matched or not, hence
\begin{equation*}
\max \big\{
u_x - \alpha_{xy},
v_y - \gamma_{xy}
\big\}
\leq \max \big\{
u_i^\pi - \alpha_{xy},
v_j^\pi - \gamma_{xy}
\big\} < 0,
\end{equation*}
which contradicts Definition \ref{def:stable_matching_deterministic}(iv).
Finally, assume towards a contradiction that there is a blocking agent $i$, so that $u_i^\pi < 0$.
We again have that $u_x \leq u_i^\pi$, hence $u_x < 0$, which contradicts Definition \ref{def:stable_matching_deterministic}(v).
By a symmetric argument, there is no blocking agent $j$ either.

\medskip

\subsection{Proof of Theorem \ref{thm:ExistenceUniquenessEquil}} ~


To prove uniqueness of a solution to the system of equation \eqref{eq:zero_excess_demand}, we show that the excess demand function $\mathbf e$ defined in \eqref{eq:excess_demand} is an M-function, hence inverse isotone and injective (see Appendix \ref{apx:Mfunctions}).

\begin{proof}[Uniqueness]
For all $x'y' \neq xy$ the functions $\boldsymbol \mu_{xy}^\alpha$ and $\boldsymbol \mu_{xy}^\gamma$ are weakly increasing in $\tau_{x'y'}$, hence the function $\mathbf e_{xy} (\tau) = \boldsymbol \mu_{xy}^\gamma (\tau^-) - \boldsymbol \mu_{xy}^\alpha (\tau^+)$ is weakly decreasing in $\tau_{x'y'}$.
The function $\mathbf e$ is therefore off-diagonally antitone.

Next, we show that $\mathbf e$ is strongly nonreversing.
Remark that under Assumption \ref{ass:nonvanishing}, the demand $\boldsymbol \mu_{x0}^\alpha (\tau^\alpha)$ for the outside option is strictly increasing in each $\tau_{xy}^\alpha$; and similarly $\boldsymbol \mu_{0y}^\gamma (\tau^\gamma)$ is strictly increasing in each $\tau_{xy}^\gamma$.
Therefore, since the aggregate excess demand $\sum_{xy} \mathbf e_{xy} (\tau) = \sum_{xy} \big( \boldsymbol \mu_{xy}^\gamma (\tau^-) - \boldsymbol \mu_{xy}^\alpha (\tau^+) \big)$ is equal to
\begin{equation*}
    \sum_{xy} \mathbf e_{xy} (\tau)
    = \sum_y m_y - \sum_x n_x
    + \sum_x \boldsymbol \mu_{x0}^\alpha (\tau^+)
    - \sum_y \boldsymbol \mu_{0y}^\gamma (\tau^-),
\end{equation*}
it is strictly increasing in $\tau_{xy}$ for all $xy$.
Now let $\tau \leq \tau'$ such that $\mathbf e (\tau) \geq \mathbf e (\tau')$, and assume towards a contradiction that $\tau \neq \tau'$.
Then $\mathbf e (\tau) \geq \mathbf e (\tau')$ implies $\sum_{xy} \mathbf e_{xy} (\tau) \geq \sum_{xy} \mathbf e_{xy} (\tau')$, and because $\tau \leq \tau'$ and $\tau \neq \tau'$, the strict monotonicity of the aggregate excess demand implies $\sum_{xy} \mathbf e_{xy} (\tau) < \sum_{xy} \mathbf e_{xy} (\tau')$, a contradiction.
In fact $\tau = \tau'$, hence $\mathbf e$ is strongly nonreversing, and thus it is an M-function.
It follows from Theorem \ref{thm:Mfunction_inverse_isotone} that $\mathbf e$ inverse isotone, hence injective.
\end{proof}

The existence proof is constructive and builds upon the inverse isotonicity of $\mathbf e$.

\begin{proof}[Existence]
We construct a sequence $(\tau^s)$ which converges to a solution of equation \eqref{eq:zero_excess_demand}.
By Assumption \ref{ass:nonvanishing}, for $c > 0$ large enough the vector $\overline \tau$ such that $\overline \tau_{xy} = c$ for all $xy$ verifies $\mathbf e (\overline \tau) \geq 0$,
and similarly, the vector $\underline \tau$ such that $\underline \tau_{xy} = -c$ for all $xy$ verifies $\mathbf e (\underline \tau) \leq 0$.
Since $\mathbf e (\underline \tau) \leq \mathbf e (\overline \tau)$, by inverse isotonicity of $\mathbf e$ we have $\underline \tau \leq \overline \tau$.
We define $\tau^0 = \underline \tau$, so that $\tau^0 \leq \overline \tau$ and $\mathbf e (\tau^0) \leq 0$, and we show that we can recursively construct $\tau^{s+1}$ such that $\tau^s \leq \tau^{s+1} \leq \overline \tau$ and $\mathbf e_{xy} (\tau^{s+1}) \leq 0$ for all $xy$.

Suppose we have constructed such a sequence up to $\tau^s$.
Then $\mathbf e_{xy} (\tau_{xy}^s, \tau_{-xy}^s) \leq 0$, and the function $\tau_{xy} \mapsto \mathbf e_{xy} (\tau_{xy}, \tau_{-xy}^s)$ is increasing and continuous.
Furthermore, since $\mathbf e$ is off-diagonally antitone and $\tau^s \leq \overline \tau$, we have $\mathbf e_{xy} (\overline \tau_{xy}, \tau_{-xy}^s) \geq \mathbf e_{xy} (\overline \tau_{xy}, \overline \tau_{-xy}) \geq 0$.
Hence by the intermediate value theorem, there exists $\tau_{xy}^{s+1}$ such that $\tau_{xy}^s \leq \tau_{xy}^{s+1} \leq \overline \tau_{xy}$ and $\mathbf e_{xy} (\tau_{xy}^{s+1}, \tau_{-xy}^s) = 0$.
Again using that $\mathbf e$ is off-diagonally antitone, we have $\mathbf e_{xy} (\tau_{xy}^{s+1}, \tau_{-xy}^{s+1}) \leq \mathbf e_{xy} (\tau_{xy}^{s+1}, \tau_{-xy}^s) = 0$ for all $xy$, i.e., $\mathbf e_{xy} (\tau^{s+1}) \leq 0$.
This proves the induction hypothesis.

Since for any $xy$, the sequence $(\tau_{xy}^s)$ is monotone nondecreasing and bounded above by $\overline \tau_{xy}$, it has a limit $\tau_{xy}^*$.
Thus, taking $s \to +\infty$ in $\mathbf e_{xy} (\tau_{xy}^{s+1}, \tau_{-xy}^s) = 0$, the continuity of $\mathbf e$ implies that $\mathbf e (\tau^*) = 0$.
Therefore $\tau^*$ solves the system \eqref{eq:zero_excess_demand}.
\end{proof}

\medskip

\subsection*{Proof of Theorem \ref{thm:limit}} ~

Let $\sigma_k = 1/k$.
For given $x$, we have $\mu_{x0} (\sigma_k) \leq n_x$ hence $-\sigma_k \ln (\mu_{x0} (\sigma_k) ) \geq -\sigma_k \ln n_x$.
The sequence $-\sigma_k \ln (\mu_{x0} (\sigma_k) )$ thus takes its values in $[-\ln n_x; +\infty)$, so up to an extraction it admits a limit $u_x^* \in [-\ln n_x; +\infty) \cup \{+\infty\}$.
Taking $k \to +\infty$ in the inequality above, we get $u_x^* \geq 0$, so in fact $u_x^* \in \R_+ \cup \{+\infty\}$.
Similarly, for any $y$ the sequence $-\sigma_k \ln \mu_{0y} (\sigma_k)$ admits (up to an extraction) a limit $v_y^* \in \R_+ \cup \{+\infty\}$.
As a result,
\begin{align*}
\max \big\{ u_x^* - \alpha_{xy}, v_y^* - \gamma_{xy} \big\}
&= \lim_{k \to +\infty} \max \big\{ -\sigma_k \ln \mu_{x0} (\sigma_k) - \alpha_{xy}, -\sigma_k \ln \mu_{0y} (\sigma_k) - \gamma_{xy} \big\} \\
&= \lim_{k \to +\infty} -\sigma_k \ln (\mu_{xy} (\sigma_k)),
\end{align*}
and since $\mu_{xy} (\sigma_k) \leq \min \{n_x, m_y\}$, by the same reasoning that limit must also be nonnegative, hence $\max \big\{ u_x - \alpha_{xy}, v_y - \gamma_{xy} \big\} \geq 0$.

Next, since the sequences $\mu_{x0} (\sigma_k)$, $\mu_{0y} (\sigma_k)$, and $\mu_{xy} (\sigma_k)$ are bounded (below by 0, above by $n_x$ or $m_y$), up to further extractions we may also define their respective limits $\mu_{x0}^*$, $\mu_{0y}^*$, and $\mu_{xy}^*$.
By continuity, these limits verify the feasibility constraints $\sum_y \mu_{xy}^* \leq n_x$ and $\sum_x \mu_{xy}^* \leq m_y$.
Now assume $\mu_{x0}^* > 0$.
Then $-\sigma_k \ln \mu_{x0} (\sigma_k) \sim -\sigma_k \ln \mu_{x0}^* \to 0$ as $k \to +\infty$, thus $u_x^* = 0$.
Similarly, $\mu_{0y}^* > 0$ implies $v_y^* = 0$.
In addition, $\mu_{xy}^* > 0$ implies $\max \big\{ u_x^* - \alpha_{xy}, v_y^* - \gamma_{xy} \big\} = 0$ from the expression above.
This also shows that $u_x^*$ and $v_y^*$ are in fact finite:
for $u_x^*$ for instance, since $\mu_{x0}^* + \sum_y \mu_{xy}^* = n_x$ we have either $\mu_{x0}^* > 0$, in which case $u_x^* = 0$, or $\mu_{xy}^* > 0$ for some $y$, in which case $u_x^* \leq \alpha_{xy}$.

Thus $(\mu^*,u^*,v^*)$ satisfies conditions (ii) to (vi) of Definition \ref{def:stable_matching_deterministic}, but not necessarily condition (i), as the integrality of $\mu_{xy}^*$ is not guaranteed.
But by the Integral flow theorem, $\mu_{xy}^*$ is a convex combination of some $K$ integral matrices $\mu^k$ that still satisfy conditions (ii) and (iii):
$\mu^* = \sum_{k=1}^K w_k \, \mu^k$, where $w_k \geq 0$ and $\sum_{k=1}^K w_k = 1$.
Thus, since $\mu_{xy}^1 > 0$ implies $\mu_{xy}^* > 0$, $\mu_{x0}^1 > 0$ implies $\mu_{x0}^* > 0$, and $\mu_{0y}^1 > 0$ implies $\mu_{0y}^* > 0$, the outcome $(\mu^1, u^*, v^*)$ satisfies conditions (i)--(vi) of Definition \ref{def:stable_matching_deterministic}.

\medskip

\subsection*{Proof of Theorem \ref{thm:constrained_choice_unique}} ~

This proof is similar to that of Theorem \ref{thm:ExistenceUniquenessEquil}.
To prove uniqueness, we show that the function $\mathbf q : \R^{|\X \times \Y|} \to \R^{|\X \times \Y|}$ defined by
\[
\mathbf q_{xy} (\theta) = \bar \mu_{xy} - \theta_{xy}^- - \boldsymbol \mu_{xy} (\theta^+)
\]
is an M-function, hence inverse isotone and injective (see Appendix \ref{apx:Mfunctions}).
The existence proof is constructive.

\begin{proof}[Uniqueness]
The function $\theta \mapsto \boldsymbol \mu_{xy} (\theta^+)$ is weakly increasing in $\theta_{x'y'}$ for $x'y' \neq xy$, hence the function $\mathbf q$ is off-diagonally antitone.
Now let us show that $\mathbf q$ is also strongly nonreversing.
Let $\theta, \theta'$ be such that $\theta \leq \theta'$ and $\mathbf q (\theta) \geq \mathbf q (\theta')$.
Then $\theta^- \geq \theta'^-$ and $\bar \mu - \theta^- - \boldsymbol \mu (\theta^+) \geq \bar \mu - \theta'^- - \boldsymbol \mu (\theta'^+)$.
By summation $\bar \mu - \boldsymbol \mu (\theta^+) \geq \bar \mu - \boldsymbol \mu (\theta'^+)$, hence $-\boldsymbol \mu (\theta^+) \geq -\boldsymbol \mu (\theta'^+)$.
Since $-\boldsymbol \mu$ is an M-function (Lemma \ref{lem:demand_inverse_isotone}) it is inverse isotone, therefore $\theta^+ \geq \theta'^+$.
However $\theta^+ \leq \theta'^+$ by assumption, hence $\theta^+ = \theta'^+$ and the inequality $\mathbf q (\theta) \geq \mathbf q (\theta')$ simplifies into $\theta^- \leq \theta'^-$.
Hence $\theta^- = \theta'^-$ as well, and therefore $\theta = \theta'$ so $\mathbf q$ is strongly nonreversing.
Applying Theorem \ref{thm:Mfunction_inverse_isotone} allows to conclude.
\end{proof}

\begin{proof}[Existence]
Recall that $\mathbf q (\theta) = \bar \mu - \theta^- - \boldsymbol \mu (\theta^+)$.
On the one hand, for $c>0$ large enough, the vector $\overline \theta$ such that $\overline \theta_{xy} = c$ for all $xy$ verifies $\overline \theta^- = 0$ and $\boldsymbol \mu_{xy} (\overline \theta^+) < \bar \mu_{xy}$ for all $xy$, hence $\mathbf q (\overline \theta) \geq 0$.
On the other hand, with $\underline \theta = -\bar \mu$ we have $\underline \theta^- = \bar \mu$ and $\boldsymbol \mu_{xy} (\underline \theta^+) > 0$ for all $xy$, hence $\mathbf q (\underline \theta) \leq 0$.
The rest of the proof consists of building a sequence $(\theta^s)$ with $\theta^0 = \underline \theta$ and such that $\theta^s \leq \theta^{s+1} \leq \overline \theta$ for all $s$, identical to the existence proof for Theorem \ref{thm:ExistenceUniquenessEquil}.
\end{proof}

\medskip

\subsection*{Proof of Theorem \ref{thm:DA_convergence}} ~

First, note that the updates described in Algorithm \ref{alg:darum} are well defined, in the sense that the capacities $\mu^{A,t}$, $\mu^{P,t}$, and $\mu^{K,t}$ are always strictly positive.
Indeed, by initialization, $\mu_{xy}^{A,0}=\min\{n_x,m_y\}>0$, hence $\mu^{P,1}=\boldsymbol c^\alpha(\mu^{A,0})$ and $\mu^{K,1}=\boldsymbol c^\gamma(\mu^{P,1})$ are also strictly positive.
Moreover, since in the proposal phase one has
$\mu^{A,0} - \mu^{P,1} = (\theta^{\alpha,1})^- \geq 0$,
it follows from the update rule that
$\mu^{A,1}
=\mu^{A,0} - (\mu^{P,1} - \mu^{K,1})
=(\theta^{\alpha,1})^- + \mu^{K,1} > 0$.
Iterating this argument shows that $\mu^{A,t}$, $\mu^{P,t}$, and $\mu^{K,t}$ remain strictly positive for all $t$.

Next, we prove a series of claims about Algorithm \ref{alg:darum} leading to the result.
Assumption \ref{ass:nonvanishing} is held throughout.

\begin{claim}  \label{claim1}
Kept offers remain in place at the next period: $\mu^{K,t} \leq \mu^{P,t+1}$.
\end{claim}

\begin{proof}
By Theorem \ref{thm:constrained_choice_unique}, the equation $\boldsymbol \mu^\alpha (\theta^+) + \theta^- = \bar \mu$ has a unique solution for any $\bar \mu$.
Let $\theta^t$ be that solution for $\bar \mu = \mu^{A,t-1}$, so that $\mu^{P,t} = \mathbf c^\alpha (\mu^{A, t-1}) = \boldsymbol \mu^\alpha ((\theta^t)^+)$.
Recall from the proof of Theorem \ref{thm:constrained_choice_unique} that for any $\bar \mu$, the function $\mathbf q (\theta) = \bar \mu - \theta^- - \boldsymbol \mu^\alpha (\theta^+)$ is an M-function, hence inverse isotone.
Taking $\bar \mu = \mu^{A, t-1}$, since $\mu^{A,t} \leq \mu^{A, t-1}$ we obtain $\mathbf q (\theta^t) = 0 \leq \mu^{A, t-1} - \mu^{A,t} = \mathbf q (\theta^{t+1})$, hence by inverse isotonicity $\theta^t \leq \theta^{t+1}$ and therefore $(\theta^t)^- \geq (\theta^{t+1})^-$.
As a result, 
\[
\mu^{A,t-1} - \mu^{P,t}
= \mu^{A,t-1} - \boldsymbol \mu^\alpha ((\theta^t)^+)
= (\theta^t)^- \geq (\theta^{t+1})^-
= \mu^{A,t} - \boldsymbol \mu^\alpha ((\theta^{t+1})^+)
= \mu^{A,t} - \mu^{P,t+1},
\]
hence $\mu^{K,t} = \mu^{A,t} - \mu^{A,t-1} + \mu^{P,t} \leq \mu^{P,t+1}$.
\end{proof}

Next, let $\tau^{\alpha,t}$ and $\tau^{\gamma,t}$ be the waiting times associated with the constrained demands $\mathbf c^\alpha (\mu^{A,t-1})$ and $\mathbf c^\gamma (\mu^{P,t})$, respectively.
We have:

\begin{claim}  \label{claim2}
$\tau^{\alpha,t}$ weakly increases with $t$, and $\tau^{\gamma,t}$ weakly decreases with $t$.
\end{claim}

\begin{proof}
We have $\tau^{\alpha,t} = (\theta^t)^+$ where $\theta^t$ is the solution to $\boldsymbol \mu^\alpha (\theta^+) + \theta^- = \mu^{A,t-1}$.
We saw in the proof of Claim \ref{claim1} that $\theta^t \leq \theta^{t+1}$, hence $\tau^{\alpha,t} \leq \tau^{\alpha,t+1}$.
For the other side of the market, we have $\tau^{\gamma,t} = (\theta^t)^+$ where $\theta^t$ is now the solution to $\boldsymbol \mu^\gamma (\theta^+) + \theta^- = \mu^{P,t}$.
The kept offers are thus $\mu^{K,t} = \mathbf c^\gamma (\mu^{P,t}) = \boldsymbol \mu^\gamma (\tau^{\gamma,t})$.
Again using that $\mathbf q (\theta) = \mu^{P,t+1} - \theta^- - \boldsymbol \mu^\gamma (\theta^+)$ is inverse isotone, and 
\[
\mathbf q (\theta^{t+1})
= 0
\leq \mu^{P,t+1} - \mu^{K,t}
= \mu^{P,t+1} - 0 - \boldsymbol \mu^\gamma (\tau^{\gamma,t})
= \mathbf q (\tau^{\gamma,t})
\]
where the inequality comes from Claim \ref{claim1}, we obtain that $\theta^{t+1} \leq \tau^{\gamma,t}$, hence $\tau^{\gamma,t+1} \leq \tau^{\gamma,t}$.
\end{proof}

\begin{claim} \label{claim3}
At every step $t$, $\min \big\{ \tau_{xy}^{\alpha,t}, \tau_{xy}^{\gamma,t} \big\} = 0$ for every $x \in \X$ and $y \in \Y$.
\end{claim}

\begin{proof}
Looking for a contradiction, assume that $\tau_{xy}^{\alpha,t} > 0$ and $\tau_{xy}^{\gamma,t} > 0$ for some $xy$ and $t$.
From Claim \ref{claim2}, $\tau_{xy}^{\gamma,t} > 0$ implies $\tau_{xy}^{\gamma, s}>0$ for all $s \in \{1,\dots,t\}$.
A positive waiting time implies a saturated constraint, so $\mu_{xy}^{K,s} = \mu_{xy}^{P,s}$ for all such $s$, which means that no offer of type $xy$ was refused up to step $t$: $\mu_{xy}^{A,s} = \mu_{xy}^{A,0} = \min \{n_x, m_y\}$.
Next, $\tau_{xy}^{\alpha,t} > 0$ implies that $\mu_{xy}^{P,t} = \mu_{xy}^{A,t-1}$, and as a result $\mu_{xy}^{K,t} = \mu_{xy}^{P,t} = \min \{n_x, m_y\}$.
This implies that either all passengers $x$ or all taxis $y$ are demanding the same option; but under Assumption \ref{ass:nonvanishing} this cannot be the case, since the distributions $\P_x$ and $\Q_y$ of the random utility shocks have nonvanishing densities.
\end{proof}

We are now ready to prove the theorem.

\begin{proof}[Proof of Theorem \ref{thm:DA_convergence}]
The sequence $\tau^{\gamma,t}$ is weakly decreasing (Claim \ref{claim2}) and bounded below by 0, so it converges; denote $\tau^{\gamma, *}$ its limit.
By continuity of $\boldsymbol \mu^\gamma$, we thus have
$\mu^{K,t} = \boldsymbol \mu^\gamma (\tau^{\gamma, t}) \to \boldsymbol \mu^\gamma (\tau^{\gamma, *})$.

Next, the sequence $\tau^{\alpha,t}$ is weakly increasing.
Looking for a contradiction, suppose that $\tau_{xy}^{\alpha,t}$ is not bounded above for some $xy$, so that $\tau_{xy}^{\alpha,t} \to \infty$.
Then $\boldsymbol \mu_{xy}^\alpha (\tau^{\alpha,t}) \to 0$, and since $\boldsymbol \mu^\alpha (\tau^{\alpha,t}) = \mu^{P,t} \geq \mu^{K,t-1}$ from Claim \ref{claim1}, taking the limit we obtain that $\boldsymbol \mu_{xy}^\gamma (\tau^{\gamma, *}) = 0$.
But this contradicts Assumption \ref{ass:nonvanishing}, which ensures a strictly positive demand for any value of $\tau^\gamma$.
In fact $\tau^{\alpha,t}$ is bounded above, so it converges; denote $\tau^{\alpha, *}$ its limit.
By continuity of $\boldsymbol \mu^\alpha$, we thus have
$\mu^{P,t} = \boldsymbol \mu^\alpha (\tau^{\alpha, t}) \to \boldsymbol \mu^\alpha (\tau^{\alpha, *})$.

Furthermore, the sequence $\mu^{A,t}$ is weakly decreasing and bounded below so it also converges, hence $\mu^{P,t} - \mu^{K,t} = \mu^{A,t-1} - \mu^{A,t} \to 0$ and therefore $\mu^{P,t}$ and $\mu^{K,t}$ have the same limit, $\mu^* := \boldsymbol \mu^\alpha (\tau^{\alpha, *}) = \boldsymbol \mu^\gamma (\tau^{\gamma, *})$, so $(\mu^*, \tau^{\alpha, *}, \tau^{\gamma, *})$ satisfies the market clearing condition.
Taking $t \to \infty$ in Claim \ref{claim3} we have $\min \big\{ \tau_{xy}^{\alpha, *}, \tau_{xy}^{\gamma, *} \big\} = 0$ for all $xy$, so $(\mu^*, \tau^{\alpha, *}, \tau^{\gamma, *})$ also satisfies the one-sided money burning condition.
Hence $(\mu^*, \tau^{\alpha, *}, \tau^{\gamma, *})$ is an aggregate stable matching with money burning in the random utility case.
Finally, note that since $\mu^* = \boldsymbol \mu^\alpha (\tau^{\alpha, *}) = \boldsymbol \mu^\gamma (\tau^{\gamma, *})$, $\tau^{\alpha, *}$ is indeed the solution to the constrained demand problem $\boldsymbol \mu^\alpha (\theta^+) + \theta^- = \mu^*$; and $\tau^{\gamma, *}$ the solution to the constrained demand problem $\boldsymbol \mu^\gamma (\theta^+) + \theta^- = \mu^*$.
\end{proof}

\end{document}